%% file: aamas.tex
\documentclass[sigconf]{aamas} 
\usepackage{balance} % for balancing columns on the final page
\usepackage{packages}
\usepackage{commands}
\usepackage{mymacros}
\setcopyright{rightsretained}
\acmConference[GAIW'22]{Appears at the 4th Games, Agents, and Incentives Workshop (GAIW 2022). Held as part of the Workshops at the 20th International Conference on Autonomous Agents and Multiagent Systems.}{May 2022}{Auckland, New Zealand}{Abramowitz, Ceppi, Dickerson, Hosseini, Lev, Mattei, Zick (Chairs)} 

\copyrightyear{2022}
\acmYear{2022}
\acmDOI{}
\acmPrice{}
\acmISBN{}

\acmSubmissionID{24}

\title{Gradient Descent Ascent in Min-Max Stackelberg Games}

\author{Denizalp Goktas}
\affiliation{
  \institution{Brown University}
  \department{Computer Science}
  \city{Providence}
  \state{Rhode Island}
  \country{USA}}
\email{denizalp_goktas@brown.edu}

\author{Amy Greenwald}
\affiliation{
  \institution{Brown University}
  \department{Computer Science}
  \city{Providence}
  \state{Rhode Island}
  \country{USA}}
\email{amy_greenwald@brown.edu}

\input{abstract}


\newcommand{\BibTeX}{\rm B\kern-.05em{\sc i\kern-.025em b}\kern-.08em\TeX}

\begin{document}

%%% The following commands remove the headers in your paper. For final 
%%% papers, these will be inserted during the pagination process.

\maketitle

\input{intro}
\input{prelim}

\input{gdad}
\input{gdadxs}
\input{experiments}

\input{conclusion}

\bibliographystyle{ACM-Reference-Format}
\bibliography{references.bib}

\newpage
\onecolumn
\appendix

\input{appendix/fisher}

\input{appendix/proofs/simGDA_proofs}
\input{appendix/proofs/gdaxs_proofs}

%%%%%%%%%%%%%%%%%%%%%%%%%%%%%%%%

\end{document}

%% file: abstract.tex
\begin{abstract}
    Min-max optimization problems (i.e., min-max games) have attracted a great deal of attention recently as their applicability to a wide range of machine learning problems has become evident. In this paper, we study min-max games with dependent strategy sets, where the strategy of the first player constrains the behavior of the second. Such games are best understood as sequential, i.e., Stackelberg, games, for which the relevant solution concept is Stackelberg equilibrium, a generalization of Nash. One of the most popular algorithms for solving min-max games is gradient descent ascent (GDA). We present a straightforward generalization of GDA to min-max Stackelberg games with dependent strategy sets, but show that it may not converge to a Stackelberg equilibrium. We then introduce two variants of GDA, which assume access to a solution oracle for the optimal Karush Kuhn Tucker (KKT) multipliers of the games' constraints. We show that such an oracle exists for a large class of convex-concave min-max Stackelberg games, and provide proof that our GDA variants with such an oracle converge in $O(\nicefrac{1}{\varepsilon^2})$ iterations to an $\varepsilon$-Stackelberg equilibrium, improving on the most efficient algorithms currently known which converge in $O(\nicefrac{1}{\varepsilon^3})$ iterations. We then show that solving Fisher markets, a canonical example of a min-max Stackelberg game, using our novel algorithm, corresponds to buyers and sellers using myopic best-response dynamics in a repeated market, allowing us to prove the convergence of these dynamics in $O(\nicefrac{1}{\varepsilon^2})$ iterations in Fisher markets. We close by describing experiments on Fisher markets which suggest potential ways to extend our theoretical results, by demonstrating how different properties of the objective function can affect the convergence and convergence rate of our algorithms.
\end{abstract}

%% file: intro.tex
\section{Introduction}
\label{sec:intro}

Min-max optimization problems (i.e., zero-sum games) have attracted a great deal of attention recently as their applicability to a wide range of machine learning problems has become evident.
Applications of min-max games include, but are not limited to, generative adversarial networks \cite{sanjabi2018convergence}, fairness in machine learning \cite{dai2019kernel, edwards2016censoring, madras2018learning, sattigeri2018fairness}, generative adversarial imitation learning \cite{cai2019global, hamedani2018iteration}, reinforcement learning \cite{dai2018rl}, adversarial learning \cite{sinha2020certifying}, and statistical learning (e.g., learning parameters of exponential families) \cite{dai2019kernel}. 

These applications often require solving \mydef{min-max games}, which are
%technically
\mydef{constrained min-max optimization problems}:
%(with independent feasible sets):
i.e.,\\
$\min_{\outer \in \outerset} \max_{\inner \in \innerset} \obj(\outer, \inner)$,
where $\obj: \outerset \times \innerset \to \R$ is continuous, and $\outerset \subset \R^\outerdim$ and $\innerset \subset \R^\innerdim$ are non-empty and compact.
In this paper, we focus on \mydef{convex-concave} min-max games, in which $\obj$ is convex in $\outer$ and concave in $\inner$.
In the special case of convex-concave objective functions, the seminal minimax theorem holds: i.e.,
$\min_{\outer \in \outerset} \max_{\inner \in \innerset} \obj(\outer, \inner) = \max_{\inner \in \innerset} \min_{\outer \in \outerset} \obj(\outer, \inner)$ \cite{neumann1928theorie}.
This theorem guarantees the existence of a saddle point, i.e., a point that is simultaneously a minimum of $\obj$ in the $\outer$-direction and a maximum of $\obj$ in the $\inner$-direction, which allows us to interpret the optimization problem as a simultaneous-move, zero-sum game, where $\inner^*$ (resp. $\outer^*$) is a best-response of the outer (resp. inner) player to the other's action $\outer^*$ (resp. $\inner^*)$, in which case a saddle point is also called a minimax point or a Nash equilibrium.

More recently, \citeauthor{goktas2021minmax} \cite{goktas2021minmax} have considered the more general problem of solving \mydef{min-max Stackelberg games}, which are constrained min-max optimization problems \mydef{with dependent feasible sets}: i.e., $\min_{\outer \in \outerset} \max_{\inner \in \innerset : \constr(\outer, \inner) \geq \zeros} \obj(\outer, \inner)$,
where $\obj: \outerset \times \innerset \to \R$ is continuous, $\outerset \subset \R^\outerdim$ and $\innerset \subset \R^\innerdim$ are non-empty and compact, and $\constr(\outer, \inner) = \left(\constr[1](\outer, \inner), \hdots, \constr[\numconstrs](\outer, \inner) \right)^T$ with $\constr[\numconstr]: \outerset \times \innerset \to \R$.
%\citeauthor{goktas2021minmax} 
These authors observe that the minimax theorem does not necessarily hold assuming dependent feasible sets \cite{goktas2021minmax}.%
\footnote{When a Nash equilibrium exists in a min-max Stackelberg game, the game reduces to a simultaneous-move game and the Stackelberg equilibrium coincides with the Nash.}
As a result, such games are more appropriately viewed as sequential, i.e., Stackelberg, games,%
\footnote{One could also view such games as pseudo-games (also known as abstract economies) \cite{facchinei2007generalized, facchinei2010generalized}, in which the players move simultaneously under the unreasonable assumption that the moves they make will satisfy the game's dependency constraints.
Under this view, the relevant solution concept is generalized Nash equilibrium.}
where the outer player chooses $\outer \in \outerset$ before the inner player responds with their choice of $\hat{\inner} (\outer) \in \innerset$ s.t.\ $\constr (\outer, \hat{\inner} (\outer)) \geq \zeros$.
In these games, the outer player's \mydef{value function}
%\footnote{\sdeni{Note that this use of the term value function in economics is distinct from its use in reinforcement learning.}{} \amy{why cross this out?}\deni{Seems obvious but if you disagree let's keep it, we just have a lot of footnotes!}}
$\val[\outerset]: \outerset \to \R$ is defined as $\val[\outerset](\outer) = \max_{\inner \in \innerset : \constr(\outer, \inner) \geq \zeros} \obj(\outer, \inner)$.
This function represents the outer player's loss, assuming the inner player chooses a feasible best response, so it is the function the outer player seeks to minimize.
The inner player's value function, $\val[\innerset]: \outerset \to \R$, which they seek to maximize, is simply the objective function given the outer player's action: i.e., $\val[\innerset](\inner; \outer) = \obj(\outer, \inner)$.

%\amy{i prefer not to mention that G\&G studied Fisher markets, or even AGT, because that kind of gives away who the authors of this current paper are.}
%\amy{can we add something about more their FOM method, and how it is not a GDA-like method, so they leave open the question of whether there is a GDA-like method that can solve games in the dependent case.}

One of the most popular algorithms for solving %min-max optimization problem with \emph{independent\/} strategy sets
min-max games is \mydef{gradient descent ascent} (GDA), which at each iteration carries out gradient descent for $\outer$, the variable being minimized, and gradient ascent for $\inner$, the variable being maximized.
Although \citeauthor{goktas2021minmax} provide first-order methods (FOMs) to solve min-max Stackelberg games with a convex-concave objective function in polynomial time, their methods are \mydef{nested GDA} algorithms, which require two nested gradient update loops.
As such, they do not have the flavor of \mydef{simultaneous GDA}, i.e., an algorithm that simultaneously runs a gradient descent step for the outer player and a gradient ascent step for the inner player.
%, and they leave open the question as to whether their methods can be adapted \samy{as such}{accordingly}.
%\sdeni{}{\citeauthor{goktas2021minmax} propose a nested FOM, which for each gradient descent step on the outer player's value function, runs an entire gradient ascent procedure on the inner player's value function.}
%\deni{Might need to add more details on conditions required.}
In this paper, we investigate the behavior of simultaneous GDA variants in min-max Stackelberg games.
%Yet, very little is known about the behavior of GDA in min-max games with dependent strategy sets, and it is not even clear what a GDA algorithm \sdeni{}{that converges to Stackelberg equilibria} would look like in min-max games with dependent strategy sets, since in the dependent strategy set setting the
The following example---which relies on \citeauthor{goktas2021minmax}'s insight that the direction of steepest descent in Stackelberg games is not given by the gradient of the objective function, but rather by the gradient of the outer player's value function---demonstrates how traditional simultaneous GDA algorithms can fail to converge to Stackelberg equilibria.

\begin{example}
\label{ex:gda-non-convergence}
Consider
the following min-max Stackelberg game: $\min_{\outer[ ] \in [-1, 1]} \max_{\inner[ ] \in [-1, 1] : 1 - (\outer[ ] + \inner[ ]) \geq 0} \outer[ ]^2 + \inner[ ] + 1 $.
The optimal solution (i.e., the Stackelberg equilibrium) of this game is given by $\outer[ ]^* = \nicefrac{1}{2}, \inner[ ]^* = \nicefrac{1}{2}$.
Consider the following simultaneous GDA algorithm,
similar to that given by \citeauthor{nedic2009gda} \cite{nedic2009gda}:%
\footnote{All notation is defined
%Section~\ref{sec:prelim}.}
in the next section.}
$$\outer[ ][\iter + 1] = \project[{\outer[ ] \in \outerset}] \left[\outer[ ][\iter] - \grad[\outer] \obj(\outer[ ][\iter], \inner[ ][\iter])\right] \enspace ,$$ $$ \inner[ ][\iter + 1] = \project[\left\{{\inner[ ] \in \innerset: \constr(\outer[ ][\iter], \inner[ ]) \geq 0} \right\}] \left[\inner[ ][\iter] + \grad[\inner] \obj(\outer[ ][\iter], \inner[ ][\iter]) \right] \enspace .$$
Applied to this game, this algorithm yields the following update rules: 
$$\outer[ ][\iter+1] = \project[{\outer[ ] \in [-1,1]}] \left[\outer[ ][\iter] - 2 \outer[ ][\iter] \right] \enspace ,$$ 
and 
$$\inner[ ][\iter + 1] = \project[\left\{{\inner[ ] \in [-1,1]: \inner[ ] \leq 1 - \outer[ ][\iter]}\right\}] \left[\inner[ ][\iter] + 1 \right]
\enspace .$$
Starting at $\outer[ ][0] = 0, \inner[ ][0] = 0$, it then proceeds as follows: $\outer[ ][1] = 0, \inner[ ][1] = 1$; $\outer[ ][2] = 0, \inner[ ][2] = 1$; and so on, thus converging to a point which is not a Stackelberg equilibrium.
Indeed, the algorithm is not even stable when initialized at the Stackelberg equilibrium.
\end{example}
In this paper, we introduce two variants of simultaneous GDA, hereafter GDA unless otherwise noted, that converge in polynomial time to Stackelberg equilibria in a large class of min-max Stackelberg games.
Our first algorithm is a deterministic algorithm whose average iterate converges in $O(\nicefrac{1}{\varepsilon^2})$ iterations to an $\varepsilon$-Stackelberg equilibrium in convex-\emph{strictly}-concave min-max Stackelberg games under standard smoothness assumptions. Our second algorithm is a randomized algorithm whose expected output converges to an $\varepsilon$-Stackelberg equilibrium in $O(\nicefrac{1}{\varepsilon^2})$ iterations in convex-concave min-max Stackelberg games under the same assumptions.
The iteration complexity of these algorithms improve on the 
%state-of-the-art
nested gradient descent ascent algorithm provided by \citeauthor{goktas2021minmax}   \cite{goktas2021minmax}, which computed an equilibrium in $O(\nicefrac{1}{\varepsilon^3})$ iterations, and thereby show that the same convergence rate as that of simultaneous GDA assuming independent strategy sets \cite{nedic2009gda} can be achieved.

Having developed GDA algorithms for the dependent strategy set setting, we conclude by using them to compute competitive equilibria in Fisher markets, which have been shown to be instances of min-max Stackelberg games \cite{goktas2021minmax}, so that their Stackelberg equilibria coincide with the competitive equilibria.
Applied to the computation of competitive equilibria in Fisher markets, our GDA algorithms correspond to myopic best-response dynamics.
In a related dynamic price-adjustment process called t\^atonnement \cite{walras}, which also converges to competitive equilibria, sellers adjust their prices incrementally, while buyers respond optimally to the sellers' price adjustments.
Our dynamics give rise to a novel t\^atonnement process in which both buyers and sellers exhibit bounded rationality \cite{simon1997models}, with sellers and buyers adjusting their prices and demands incrementally, respectively.

Finally, we run experiments with our randomized GDA algorithm which suggest that our randomized algorithm can be derandomized, because in all our experiments the average iterates converge to a competitive equilibrium.
Our experiments 
%reveal how different smoothness properties affect the convergence rate of our algorithm, and
suggest avenues for future work, namely investigating how varying the degree of smoothness of the objective function can impact convergence rates.

\paragraph{Related Work}

Our model of min-max Stackelberg games seems to have first been studied by \citeauthor{wald1945maximin}, under the posthumous name of Wald's maximin model~\cite{wald1945maximin}.
%\amy{are you sure it is a simplification, and that the RO model is not equivalent to ours?}\deni{Here is why I am not feeling good about saying equivalent, if I can solve either of the two problems, then I can use the algorithm to solve one of these problems to solve the other with some additional work. However, mathematically, Wald's model is more general than RO since it solves for $\outer$ and $\inner$ both while RO solves only for $\inner$.}
A variant of Wald's maximin model is the main paradigm used in robust optimization, a fundamental framework in operations research for which many methods have been proposed \cite{ben2015oracle, ho2018online, postek2021firstorder}. 
%but these algorithms do not solve for a Stackelberg equilibrium, transform the problem into a feasibility problem, and
\citeauthor{shimizu1980minmax}  \cite{shizimu1981stackelberg,shimizu1980minmax} proposed the first algorithm to solve min-max Stackelberg games
%\amy{Wald's maximin problem? were they solving RO?} \deni{no min-max Stackelberg games actually}
via a relaxation to a constrained optimization problem with infinitely many constraints, which nonetheless seems to perform well in practice.
%its run time in the worst case is superexponential, as it involves solving an optimization problem with infinitely many constraints in the worst case.
More recently, \citeauthor{segundo2012evolutionary} \cite{segundo2012evolutionary} proposed an evolutionary algorithm for these games, but they provided no guarantees.
As pointed out by \citeauthor{postek2021firstorder}, all prior methods either require oracles and are stochastic in nature \cite{ben2015oracle}, or rely on a binary search for the optimal value, which can be computationally complex \cite{ho2018online}.
The algorithms we propose in this paper circumvent the aforementioned issues and can be used to solve a large class of convex robust optimization problems in a simple and efficient manner.

Much progress has been made recently in solving min-max games with independent strategy sets, both in the convex-concave case and in the non-convex-non-concave case. 
For the former case, when $\obj$ is $\mu_\outer$-strongly-convex in $\outer$ and $\mu_\inner$-strongly-concave in $\inner$, gradient-descent-ascent (GDA)-based methods, compute a solution in $\tilde{O}(\mu_\inner + \mu_\outer)$ iterations \cite{mokhtari2020convergence, tseng1995variational, nesterov2006variational, gidel2020variational, lin2020near, alkousa2020accelerated} and $\tilde{\Omega}(\sqrt{\mu_\inner \mu_\outer})$ iterations \cite{ibrahim2019lower, zhang2020lower, lin2020near, alkousa2020accelerated}.
For the special case where $\obj$ is $\mu_\outer$-strongly convex in $\outer$ and linear in $\inner$ there exist methods that converge to an $\varepsilon$-approximate solution in $O(\sqrt{\nicefrac{\mu_\outer}{\varepsilon}})$ iterations \cite{juditsky2011first, hamedani2018primal, zhao2019optimal}.
When the strong concavity or linearity assumptions of $\obj$ on $\inner$ are dropped, and
$\obj$ is assumed to be $\mu_\outer$-strongly-convex in $\outer$ but only concave in $\inner$, there exist methods that converge to an $\varepsilon$-approximate solution in $\tilde{O}(\nicefrac{\mu_\outer}{\varepsilon})$ iterations \cite{ouyang2018lower, zhao2019optimal, lin2020near} with a lower bound of $\tilde{\Omega}\left(\sqrt{\nicefrac{\mu_\outer}{\varepsilon}}\right)$ iterations.
When $\obj$ is simply assumed to be convex-concave, a multitude of well-known first order methods exist that solve for an $\varepsilon$-approximate solution with $\tilde{O}\left(\varepsilon^{-1}\right)$ iteration complexity \cite{nemirovski2004prox, nesterov2007dual, tseng2008accelerated}.
When $\obj$ is assumed to be non-convex-$\mu_\inner$-strongly-concave, there exist methods that converge to various solution concepts, all in $O(\varepsilon^{-2})$ iteration \cite{sanjabi2018stoch, jin2020local, lin2020gradient, rafique2019nonconvex, lu2019block}.
%\deni{We use the word local for Stackelberg and that definition is indeed a local definition, but first-order Nash is not really local, it is a bit weirder than that, i.e., it is a Nash eqm in which players cannot make improvement in utility by only making deviations based on first order information *about the objective function*.} \amy{we have local minimax in the last paragraph of the conclusion.} \deni{Local minimax is actually local Stackelberg!!! A minimax point is what Jordan calls Stackelberg equilibria and then Dai et al. introduce a local version of minimax, which is essentially local Stackelberg! So for them when you say a minimax point they just mean a solution to the min-max problem but not necessarily to the max-min problem. Note that their definition is different than Costis' where costs defines a minimax point as a saddle point}.
%\denit{A first-order Nash is a stationary point of the objective function for both players, while a local Stackelberg eqm is a point for which the value function's gradient is 0 for the min player and the gradient of the objective is 0 for the max player. The difference between the two are basically simultaneous game vs. sequential game.}
%
When $\obj$ is non-convex-non-concave various algorithms have been proposed to compute first-order Nash equilibrium \cite{lu2019block, nouiehed2019solving}, with at best an upper bound of $\tilde{O}\left(\varepsilon^{-2.5}\right)$ iterations \cite{lin2020near, ostrovskii2020efficient}.
When $\obj$ is non-convex-non-concave and the desired solution concept is a ``local'' Stackelberg equilibrium, there exist many algorithms to compute a solution \cite{jin2020local, rafique2019nonconvex, lin2020gradient}, with the most efficient ones converging to an $\varepsilon$-approximate solution in $\tilde{O}\left( \varepsilon^{-3}\right)$ iterations \cite{thekumparampil2019efficient, zhao2020prim, lin2020near}.

% The study of competitive equilibria computation in Fisher markets was initiated by \citeauthor{devanur2002market} \cite{devanur2002market}, who provided a polynomial-time method for solving these markets assuming linear utilities.
% \citeauthor{jain2005market} \cite{jain2005market} subsequently showed that a large class of Fisher markets could be solved in polynomial-time using interior point methods.
% Recently, \citeauthor{gao2020polygm} \cite{gao2020polygm} proposed first-order methods for solving Fisher markets (only; not general min-max Stackelberg games), assuming linear, quasilinear, and Leontief utilities, as such methods can be more efficient when markets are large.
% More recently, \citeauthor{fisher-tatonnement} \citeyear{fisher-tatonnement} provided convergence results for t\^atonnement in constant-elasticity-of-substitution markets \cite{mas-colell}. 

%To the best of our knowledge, algorithms designed to solve min-max Stackelberg games are few, and are not accompanied by polynomial-time computation guarantees \sdeni{}{\cite{sniedovich2008wald}}.
%\amy{i just commented out a sentence and a citation \cite{sniedovich2008wald}; should the citation go somewhere else?}\deni{I now think we should remove it because even though it is a reference that that says that there exists no methods to solve Wald's maximin problem, it is from 2008 so maybe not old enough, and it is not published anywhere of importance (or rather I am not sure it is peer-reviewed).}

%\if 0
Extensive-form games in which players' strategy sets can depend on other players' actions have been studied by \citeauthor{davis2019solving} \cite{davis2019solving} assuming payoffs are bilinear, and by \citeauthor{farina2019regret} \cite{farina2019regret} for another specific class of convex-concave payoff functions.
\citeauthor{fabiani2021local} \cite{fabiani2021local} and \citeauthor{kebriaei2017discrete} \cite{kebriaei2017discrete} study more general settings than ours, namely non-zero-sum Stackelberg games with more than two players.
Both sets of authors derive convergence guarantees
%%% LOCAL SE is undefined
%to local Stackelberg equilibria 
assuming specific payoff structures, but their algorithms do not converge in polynomial time.
%\fi

%%% SAVE
% \samy{\citeauthor{fabiani2021local} \cite{fabiani2021local} provide an algorithm to solve for a local Stackelberg equilibrium of non-zero-sum Stackelberg games with more than two players. The authors assume a specific payoff structure to derive convergence results for their algorithms. As opposed to our setting, the authors are not focused on zero-sum games, and as a result, global Stackelberg equilibria cannot be found efficiently, which is why they consider local Stackelberg equilibria, for which nevertheless, they do not provide polynomial time computation guarantees. \citeauthor{kebriaei2017discrete} \cite{kebriaei2017discrete} study a dynamic non-zero-sum Stackelberg game with one leader, one disturbance, and $n$ players, in which the followers’ strategies can also affect each others’ payoffs. Once again this is a more general model than ours, although similar to \citeauthor{fabiani2021local} \cite{fabiani2021local}, the authors assume a specific payoff structure to derive convergence results for their algorithms, but do not provide polynomial-time guarantees.}{}

Min-max Stackelberg games naturally model various economic settings. 
They are related to abstract economies, first studied by \citeauthor{arrow-debreu} \cite{arrow-debreu}; however, the solution concept that has been the focus of this literature is generalized Nash equilibrium \cite{facchinei2007generalized, facchinei2010generalized}, which, like Stackelberg, is a weaker solution concept than Nash, but which makes the arguably unreasonable assumption that the players move simultaneously and nonetheless satisfy the constraint dependencies on their strategies imposed by one another's moves.

%Without using our precise language,
%\citeauthor{duetting2019optimal} \cite{duetting2019optimal} study optimal auction design problems.
%They propose a neural network architecture called RegretNet that represents optimal auctions, and train their networks using Algorithm~\ref{ngd-thms}.
Optimal auction design problems can be seen as min-max Stackelberg games.
\citeauthor{duetting2019optimal} \cite{duetting2019optimal} propose a neural network architecture called RegretNet to solve auctions; however, as their objectives can be non-convex-concave,
%in general, 
our guarantees do not apply.
%\amy{i think we might want to say here, or in future work, that the success of their work suggests that their may be cases of benign non-convexity lurking in the space of optimal auction design!}\deni{I am not sure that will be clear by stating that here because we have no proof that our algorithm would perform well for benign non-convex problems and we haven't even defined benign non-convex problems.}\amy{so please add something to the conclusion then!}

%Hence, optimal auction design can be seen as a (non-convex-concave) min-max Stackelberg game.
In this paper, we observe that solving for the competitive equilibrium of a Fisher market can also be seen as solving a (convex-concave) min-max Stackelberg game.
The study of the computation of competitive equilibria in Fisher markets was initiated by \citeauthor{devanur2002market} \cite{devanur2002market}, who provided a polynomial-time method for the case of linear utilities.
\citeauthor{jain2005market} \cite{jain2005market} subsequently showed that a large class of Fisher markets could be solved in polynomial-time using interior point methods.
Recently, \citeauthor{gao2020polygm} \cite{gao2020polygm} studied an alternative family of first-order methods for solving Fisher markets (only; not min-max Stackelberg games more generally), assuming linear, quasilinear, and Leontief utilities, as such methods can be more efficient when markets are large.

% A related problem setting to that of min-max games with dependent strategy sets is extensive-form games in which players' strategy sets can depend on other players' actions.
% Such games have been studied by \citeauthor{davis2019solving} \cite{davis2019solving} assuming payoffs are bilinear, and by \citeauthor{farina2019regret} \cite{farina2019regret} for another specific class of convex-concave payoffs. Several algorithms have been proposed to solve convex-concave min-max Stackelberg games \cite{segundo2012evolutionary, shizimu1981stackelberg, shimizu1980minmax}.
% Recently, \citeauthor{goktas2021minmax} 
% \cite{goktas2021minmax} proposed the first such algorithms with polynomial-time guarantees.
% Further, min-max Stackelberg games are mathematically related to the abstract economies first studied by \citeauthor{arrow-debreu} \cite{arrow-debreu}; however, the solution concept that has been the focus of most work in this literature is generalized Nash equilibrium \cite{facchinei2007generalized, facchinei2010generalized}, which, like Stackelberg, is a weaker solution concept than Nash equilibrium, but which makes the arguably unreasonable assumption that the players move simultaneously and nonetheless satisfy the constraint dependencies on their strategies imposed by one anothers' moves.

%% file: prelim.tex
\section{Preliminaries}
\label{sec:prelim}

\paragraph{Notation}

We use Roman uppercase letters to denote sets (e.g., $X$),
% calligraphic uppercase letters to denote correspondences (e.g., $\calX$),\amy{do we still have any correspondences? we might not need this.}
bold uppercase letters to denote matrices (e.g., $\allocation$), bold lowercase letters to denote vectors (e.g., $\price$), and Roman lowercase letters to denote scalar quantities, (e.g., $c$).
We denote the $i$th row vector of a matrix (e.g., $\allocation$) by the corresponding bold lowercase letter with subscript $i$ (e.g., $\allocation[\buyer])$.
Similarly, we denote the $j$th entry of a vector (e.g., $\price$ or $\allocation[\buyer]$) by the corresponding Roman lowercase letter with subscript $j$ (e.g., $\price[\good]$ or $\allocation[\buyer][\good]$).
We denote the vector of ones of size $\numbuyers$ by $\ones[\numbuyers]$.
We denote the set of integers $\left\{1, \hdots, n\right\}$ by $[n]$, the set of natural numbers by $\N$, the set of positive natural numbers by $\N_+$ the set of real numbers by $\R$, the set of non-negative real numbers by $\R_+$, and the set of strictly positive real numbers by $\R_{++}$.
We denote the orthogonal projection operator onto a convex set $C$ by $\project[C]$, i.e., $\project[C](\x) = \argmin_{\y \in C} \left\|\x - \y \right\|^2$.

\paragraph{Problem Definition}

A \mydef{min-max game with dependent strategy sets}, denoted $(\outerset, \innerset, \obj, \constr)$, is a two-player, zero-sum game, where one player, who we call the outer, or $\outer$, (resp.\ inner, or $\inner$) player, is trying to minimize their loss (resp.\ maximize their gain), defined by a continuous \mydef{objective function} $\obj: X \times Y \rightarrow \mathbb{R}$, by taking an action from their \mydef{strategy set} $\outerset \subset \R^\outerdim$ (resp. $\innerset \subset \R^\innerdim$) s.t.\ $\constr(\outer, \inner) \geq 0$, where $\constr(\outer, \inner) = \left(\constr[1](\outer, \inner), \hdots, \constr[\numconstrs](\outer, \inner) \right)^T$ with $\constr[\numconstr]: \outerset \times \innerset \to \R$.
A strategy profile $(\outer, \inner) \in \outerset \times \innerset$ is said to be \mydef{feasible} iff for all $\numconstr \in [\numconstrs]$, $\constr[\numconstr](\outer, \inner) \geq 0$.
The function $\obj$ maps a pair of feasible actions taken by the players $(\outer, \inner) \in \outerset \times \innerset$ to a real value (i.e., a payoff), which represents the loss (resp.\ the gain) of the outer (resp.\ inner) player.
A min-max game is said to be convex-concave if the objective function $\obj$ is convex-concave.

One way to see this game is as a \mydef{Stackelberg game}, i.e., a sequential game with two players, where WLOG, we assume that the minimizing player moves first and the maximizing player moves second.
The relevant solution concept for Stackelberg games is the \mydef{Stackelberg equilibrium}:
%\begin{definition}[Stackelberg Equilibrium]
%Consider the min-max game with dependent strategy sets $(\outerset, \innerset, \obj, \constr)$.
A strategy profile $\left(\outer^{*}, \inner^{*} \right) \in \outerset \times \innerset$ s.t.\ $\constr(\outer^{*}, \inner^{*}) \geq \zeros$ is an $(\epsilon, \delta)$-Stackelberg equilibrium if
%\begin{equation}
$$\max_{\inner \in \innerset : \constr(\outer^{*}, \inner) \geq 0} \obj \left( \outer^{*}, \inner \right) - \delta \leq \obj \left( \outer^{*}, \inner^{*} \right) \leq \min_{\outer \in \outerset} \max_{\inner \in \innerset : \constr(\outer, \inner) \geq 0} \obj \left( \outer, \inner \right) + \epsilon
\enspace .$$
%\end{equation}
%\end{definition}
%
%\noindent
Intuitively, an $(\varepsilon, \delta)$-Stackelberg equilibrium is a point at which the outer (resp.\ inner) player's payoff is no more than $\varepsilon$ (resp.\ $\delta$) away from its optimum.
A $(0,0)$-Stackelberg equilibrium is guaranteed to exist in min-max Stackelberg games \cite{goktas2021minmax}.
Note that when $\constr (\outer, \inner) \geq 0$ for all $(\outer, \inner) \in \outerset \times \innerset$, the game reduces to a min-max game,%
\footnote{We use the terminology ``min-max game'' to mean a min-max game with independent strategy sets, and ``min-max Stackelberg game'' to mean a min-max game with dependent strategy sets.}
for which, by the min-max theorem, a Nash equilibrium exists \cite{neumann1928theorie}.

\textbf{Mathematical Preliminaries}
%We define several concepts that are needed in our convergence proofs.
Given $A \subset \R^\outerdim$, the function $\obj: A \to \R$ is said to be $\lipschitz[\obj]$-\mydef{Lipschitz-continuous} iff $\forall \outer_1, \outer_2 \in X, \left\| \obj(\outer_1) - \obj(\outer_2) \right\| \leq \lipschitz[\obj] \left\| \outer_1 - \outer_2 \right\|$.
If the gradient of $\obj$, $\grad \obj$, is $\lipschitz[\grad \obj]$-Lipschitz-continuous, we refer to $\obj$ as $\lipschitz[\grad \obj]$-\mydef{Lipschitz-smooth}.
A function $\obj: A \to \R$ is $\mu$-\mydef{strongly convex} if $\obj(\outer_1) \geq \obj(\outer_2) + \left< \grad[\outer] \obj(\outer_2), \outer_1 - \outer_2 \right> + \nicefrac{\mu}{2} \left\| \outer_1 - \outer_1 \right\|^2$, and $\mu$-\mydef{strongly concave} if $-\obj$ is $\mu$-strongly convex.

%% file: gdad.tex
\section{(Simultaneous) GDA}\label{sec:GDA}

In this section, we explore 
%simultaneous
GDA algorithms for min-max Stackelberg games.
We prove that a min-max Stackelberg game with dependent strategy sets is equivalent to a three-player game with independent strategy sets, with the Lagrangian as the payoff function and a third player who chooses the optimal Karush Kuhn Tucker (KKT) multipliers \cite{kuhn1951kkt}.
This theorem provides a straightforward path to generalizing 
%simultaneous 
GDA to the dependent strategy sets setting. All the results in this paper are based on the following assumptions:
\begin{assumption}
\label{main-assum}
1. (\emph{Slater's condition} \cite{slater1959convex, slater2014convex})~$\forall \outer \in \outerset, \exists \widehat{\inner} \in \innerset$ s.t.\ $g_{\numconstr}(\outer, \widehat{\inner}) > \zeros$, for all $\numconstr \in [\numconstrs]$;
2.~$\obj, \constr[1], \ldots, \constr[\numconstrs]$ are continuous and convex-concave; and 3.~$\grad[\outer] f, \grad[\outer] \constr[1], \ldots, \grad[\outer] \constr[\numconstrs]$ are 
% well-defined for all $(\outer, \inner) \in \outerset \times \innerset$ and
continuous in $(\outer, \inner)$.
\end{assumption}

We note that these assumptions are in line with previous work geared towards solving min-max Stackelberg games with first-order methods (FOMs) \cite{goktas2021minmax}.
Part 1 of \Cref{main-assum},
Slater's condition, is a constraint qualification condition which is necessary to derive the optimality conditions for the inner player's payoff maximization problem.
This condition is a standard constraint qualification in the convex programming literature \cite{boyd2004convex}; without it the problem becomes analytically intractable.
Part 2 of \Cref{main-assum} is a standard assumption in exploratory
%\amy{what do you mean, exploratory?}
studies of min-max games (e.g., \cite{nouiehed2019solving}).
Finally, we note that Part 3 of \Cref{main-assum} can be replaced by a subgradient boundedness assumption; however, for simplicity, we assume this stronger condition.

Building on ideas developed by \citeauthor{nouiehed2019solving} \cite{nouiehed2019solving} and \citeauthor{jin2020local} \cite{jin2020local} for min-max games,
%the independent strategy set setting,
\citeauthor{goktas2021minmax} \cite{goktas2021minmax} recently FOMs that solve min-max Stackelberg games.
%with dependent strategy sets
%\amy{we cannot say that we proposed the FIRST FOM, can we. the others might have also proposed FOMs (not sure?)} \deni{I think we might be the first ones actually? but either way it is better we do not say that!}
%Their algorithm\amy{we have two algorithms!} is a nested gradient descent ascent algorithm, which requires two nested gradient update loops.
However, the authors leave open the problem of developing a (simultaneous) gradient descent ascent (GDA) algorithm
%, i.e., an algorithm that \samy{}{simultaneously} runs a gradient descent step for the outer player and a gradient ascent step for the inner player, 
%is of interest both for problems in machine learning and economics.
with polynomial-time guarantees.
\Cref{ex:gda-non-convergence} demonstrates that the 
%simultaneous
GDA algorithms used to solve min-max games \cite{daskalakis2018limit,lin2020gradient,nedic2009gda} do not solve min-max Stackelberg games.
%\amy{does it show that all of them do not, or just the one?}\deni{There is only one GDA family of algorithms in the independent setting (the parameters of this family are the learning rates).}
Nonetheless, we resolve this open question by developing a 
%simultaneous
GDA algorithm that approximates Stackelberg equilibria.

%For the rest of this paper, 
We define $\lang[\outer](\inner, \langmult) = \obj(\outer, \inner) + \sum_{\numconstr = 1}^\numconstrs \langmult[\numconstr] \constr[\numconstr](\outer, \inner)$ to be the Lagrangian associated with the outer player's value function, or equivalently, the inner player's payoff maximization problem, given the outer player's strategy $\outer \in \outerset$.
Our algorithms rely heavily on the next theorem, proven in \Cref{sec:app-GDA-proofs}, which states that any min-max Stackelberg game with dependent strategy sets is equivalent to a three-player min-max-min game with independent strategy sets and payoff function $\lang[\outer](\inner, \langmult)$, where the inner two players, $\inner$ and $\langmult$, can move simultaneously (i.e., they play a Nash equilibrium), but the outer player, $\outer$, must move first.

\begin{theorem}
\label{thm:stackelberg-equiv}
Under \Cref{main-assum}, any min-max Stackelberg game with dependent strategy sets $(\outerset, \innerset, \obj, \constr)$
%Given the outer player's action $\outer$, let $\lang[\outer](\inner, \langmult) = \obj(\outer, \inner) + \sum_{\numconstr = 1}^\numconstrs \langmult[\numconstr] \constr[\numconstr](\outer, \inner)$ be the Lagrangian associated with the inner player's payoff maximization problem: i.e., $\max_{\inner \in \innerset : \constr(\outer, \inner) \geq \zeros} \obj(\outer, \inner)$.
can be viewed as a min-max game with independent strategy sets: i.e.,
%
%\begin{align}
\begin{align}
    \min_{\outer \in \outerset} \max_{\inner \in \innerset : \constr(\outer, \inner) \geq \zeros} \obj(\outer, \inner) 
    &= \min_{\outer \in \outerset} \max_{\inner \in \innerset } \min_{\langmult \geq \zeros} \lang[\outer](\inner, \langmult) \\
    &= \min_{\langmult \geq \zeros} \min_{\outer \in \outerset}  \max_{\inner \in \innerset }  \lang[\outer](\inner, \langmult) \enspace .
\end{align}
%\enspace .
%\end{align} 
\end{theorem}
%
%
\if 0
\begin{corollary}
\label{corr:min-max-stackelberg-as-simult}
Assume as in \Cref{thm:stackelberg-equiv}.
Then
%
%\begin{align}
$\min_{\outer \in \outerset} \max_{\inner \in \innerset : \constr(\outer, \inner) \geq \zeros} \obj(\outer, \inner) =  \min_{\langmult \geq \zeros, \outer \in \outerset} \max_{\inner \in \innerset } \lang[\outer]( \inner, \langmult)$. %\enspace .
%\end{align}
\end{corollary}
\fi
%
%
Note, that $\min_{\langmult \geq \zeros, \outer \in \outerset} \max_{\inner \in \innerset } \lang[\outer]( \inner, \langmult)$ is a non-convex-concave min-max game, since $\lang[\outer]( \inner, \langmult)$ is non-convex-concave, due to the term $\sum_{\numconstr = 1}^\numconstrs \langmult[\numconstr] \constr[\numconstr](\outer, \inner)$, which is not guaranteed to be convex in $\langmult$ and $\outer$ jointly.\footnote{We recall that a function $\obj: A \times B \to \R$ is jointly convex in $a$ and $b$ if it is convex in $(a, b)$, and that the product of two convex functions need not give rise to a convex function.
In particular, $\langmult \cdot \constr(\outer, \inner)$ is not necessarily jointly convex in $(\langmult, \outer)$.}
%\deni{\href{https://math.stackexchange.com/questions/738676/convexity-definition-confusion}{here} is an explanation.}
%%%\amy{see stack exchange example. the Langrangian is convex concave convex, b/c it is convex in $x$, concave in $y$, and convex in $\lambda$, individually. but it is not convex in $x$ and $\lambda$ jointly.}
This reduction of a min-max game with dependent strategy sets to one with independent strategy sets does not imply that the game is solvable in polynomial time, since for non-convex-concave min-max games the computation of (global) minimax points
%\deni{I used this but havent defined it anywhere, is this fine?}\amy{YES!}
is NP-Hard \cite{daskalakis2020complexity}.
Nonetheless, \Cref{thm:stackelberg-equiv} is suggestive of a new algorithm, which we call \mydef{gradient descent descent ascent (G2DA)} (\Cref{alg:g2da}), which is essentially 
%simultaneous 
GDA, but run on the two-player non-convex-concave min-max game with independent strategy sets $\min_{\langmult \geq \zeros, \outer \in \outerset} \max_{\inner \in \innerset} \lang[\outer]( \inner, \langmult)$.

\begin{algorithm}[ht]
\caption{Gradient Descent Descent Ascent (G2DA)}
\label{alg:g2da}
\textbf{Inputs:} $\langmults, \outerset, \innerset, \obj, \constr, \learnrate[][\langmult], \learnrate[][\outer], \learnrate[][\inner], \iters, \langmult^{(0)}, \outer^{(0)}, \inner^{(0)}$ \\ 
\textbf{Output:} $\outer^{*}, \inner^{*}$
\begin{algorithmic}[1]
\For{$\iter = 1, \hdots, \iters -1$}
    \State Set $\langmult^{(\iter + 1)} = \project[\R_+] \left( \langmult^{(\iter)} - \learnrate[\iter][\langmult] \grad[\langmult] \lang[{\outer[][\iter]}](\inner[][\iter], \langmult[][\iter]) \right)$
    
    \State Set $\outer^{(\iter +1)} = \project[\outerset] \left( \outer^{(\iter)} - \learnrate[\iter][\outer] \grad[\outer] \lang[{\outer[][\iter]}](\inner[][\iter], \langmult[][\iter])  \right)$
    
    \State Set $\inner^{(\iter +1)} = \project[{
    % \inner \in 
    \innerset
    % : \constr (\outer^{(\iter)}, \inner) 
    }] \left( \inner^{(\iter)} + \learnrate[\iter][\inner] \grad[\inner] \lang[{\outer[][\iter]}](\inner[][\iter], \langmult[][\iter])  \right)$
\EndFor
\State \Return $(\outer^{(\iters)}, \inner^{(\iters)})$
\end{algorithmic}
\end{algorithm}

Unfortunately, 
% as the next example shows
as shown next
, G2DA
%\Cref{alg:g2da} 
does not converge to Stackelberg equilibria in general.
\begin{example}
\label{ex:g2da-non-convergence}
\if 0
Consider the min-max Stackelberg game with dependent strategy sets:
%
\begin{align}
    \min_{\outer[ ] \in [-1, 1]} \max_{\inner[ ] \in [-1, 1] : 0 \leq 1 - (\outer[ ] + \inner[ ])} \outer[ ]^2 + \inner[ ] + 1 \enspace .
\end{align}
\fi
Recall \Cref{ex:gda-non-convergence}.
The Lagrangian associated with the outer player's value function is $\lang(\outer[ ], \inner[ ], \langmult[ ]) = \outer[ ]^2 + \inner[ ] + \langmult[ ] (1 - (\outer[ ] + \inner[ ]))$.
The optimal solution to this game is $\langmult[ ]^* = 1, \outer[ ]^* = \nicefrac{1}{2}, \inner[ ]^* = \nicefrac{1}{2}$.
G2DA
%\Cref{alg:g2da}
applied to this game with learning rates $\learnrate[\iter][\langmult] = \learnrate[\iter][\outer] = \learnrate[\iter][\inner] = 1$ for all $\iter \in \N_+$ yields:
    $\langmult[ ][\iter + 1] = 
    % \project[\R_+]\left[ \langmult[ ][\iter] - \learnrate[\iter][\langmult] ( 1- (\outer[ ][\iter] + \inner[ ][t]))\right]  =
    \project[\R_+] \left[ \langmult[ ][\iter] - 1 +  \outer[ ][\iter] + \inner[ ][t]   \right]$,
    $\outer[ ][\iter+1] = 
    % \project[{[-1,1]}]\left[ \outer[ ][\iter] - \learnrate[\iter][\outer](2 \outer[ ][\iter] + \langmult[ ][\iter])\right] = 
    \project[{[-1,1]}]\left[ -  \outer[ ][\iter] + \langmult[ ][\iter]\right]$, and
    $\inner[ ][\iter+1] = 
    % \project[{[-1,1]}]\left[ \inner[ ][\iter] + \learnrate[\iter][\inner](1 - \langmult[ ][\iter])\right] =
    \project[{[-1,1]}]\left[ \inner[ ][\iter] + 1 - \langmult[ ][\iter]\right]$. Starting at $\langmult[ ][0] = 0, \outer[ ][0] = 0, \inner[ ][0] = 0$, the algorithm proceeds as follows $\langmult[ ][1] = 0, \outer[ ][1] = 0, \inner[ ][1] = 1$; $\langmult[ ][2] = 0, \outer[ ][2] = 0, \inner[ ][2] = 1$; and so on, and thus does not converge to a Stackelberg equilibrium of the game.
\end{example}
This result is slightly discouraging, however, not entirely surprising, since one timescale GDA, i.e., GDA run with the same step sizes while descending and ascending, with fixed step sizes does not necessarily converge to stationary points of the outer player's value function in non-convex-concave min-max games with independent strategy sets \cite{daskalakis2018limit}.
Recently, it has been shown that two timescale GDA \cite{lin2020gradient}, i.e., GDA run with different step sizes while descending and ascending, converges to stationary points of the outer player's value function in these games.
%in non-convex-concave min-max games with independent strategy sets.
Hence, one might wonder whether a two timescale variant of \Cref{alg:g2da} would converge to a Stackelberg equilibrium.
Even more discouraging, however, is that the above example converges to a solution that does not correspond to a Stackelberg equilibrium for any learning rate $\learnrate[\iter][\outer], \learnrate[\iter][\inner], \learnrate[\iter][\langmult] > 0$.
But all hope is not lost; in what follows, we rely on an oracle, inspired by the algorithms in \citeauthor{goktas2021minmax} \cite{goktas2021minmax}.

%% file: gdadxs.tex
\section{GDA with Optimal KKT Multipliers}
\label{sec:oracle}
We prove the main theoretical result of this paper in this section.
For a large class of min-max Stackelberg games, we provide first-order methods with polynomial-time guarantees which, in contrast to \citeauthor{goktas2021minmax} \cite{goktas2021minmax}, who developed a nested GDA algorithm (i.e., one with a nested loop), require only a single loop of gradient evaluations.
To do so, we rely on a Lagrangian oracle, which outputs the optimal KKT multipliers of the Lagrangian associated with the inner player's payoff maximization problem.
We also provide sufficient conditions for this oracle to exist.

%Let $\lang[\outer](\inner, \langmult) = \obj(\outer, \inner) + \sum_{\constr = 1}^\numconstrs \langmult[\constr] \constr[\numconstr](\outer, \inner)$ be the Lagrangian associated with the inner player's payoff maximization problem, given the outer player's strategy $\outer \in \outerset$.
Since the Lagrangian is convex-concave in $\outer$ and $\inner$, if the optimal KKT multipliers $\langmult^* \in \R^\numconstrs$ were known for the problem
$
    \min_{\outer \in \outerset} \max_{\inner \in \innerset : \constr(\outer, \inner) \geq \zeros} \obj(\outer, \inner) = \min_{\langmult \geq \zeros, \outer \in \outerset} \max_{\inner \in \innerset } \lang[\outer]( \inner, \langmult)
$,
then one could plug them back into the Lagrangian to obtain a convex-concave saddle point problem given by
$
    \min_{\outer \in \outerset} \max_{\inner \in \innerset }\\ \lang[\outer]( \inner, \langmult^*)
$.
This might lead one to think that they can use GDA \cite{nedic2009gda} to solve the Lagrangian for $\outer$ and $\inner$ giving rise to \mydef{Lagrangian Gradient Descent Ascent} (\mydef{LGDA}, \Cref{alg:vgda}) 
for any convex-concave min-max Stackelberg game.

\begin{algorithm}[ht]
\caption{Lagrangian Gradient Descent Ascent (LGDA)}
\label{alg:vgda}
\textbf{Inputs:} $\langmult^*, \outerset, \innerset, \obj, \constr, \learnrate[][\langmult], \learnrate[][\outer], \learnrate[][\inner], \iters, \langmult^{(0)}, \outer^{(0)}, \inner^{(0)}$ \\ 
\textbf{Output:} $\outer^{*}, \inner^{*}$
\begin{algorithmic}[1]
\For{$\iter = 1, \hdots, \iters -1$}    
    \State Set $\outer^{(\iter +1)} = \project[\outerset] \left( \outer^{(\iter)} - \learnrate[\iter][\outer] \grad[\outer] \lang[{\outer[][\iter]}](\inner[][\iter], \langmult^*)
    % \left[\grad[\outer] \obj (\outer^{(\iter)}, \inner^{(\iter)}) + \sum_{\numconstr = 1}^\numconstrs \langmult[\numconstr]^* \grad[\outer]  \constr[\numconstr](\outer^{(\iter)}, \inner^{(\iter)}) \right] 
    \right)$

    \State Set $\inner^{(\iter +1)} = \project[{
    \innerset
    }] \left( \inner^{(\iter)} + \learnrate[\iter][\inner] \grad[\inner] \lang[{\outer[][\iter]}](\inner[][\iter], \langmult^*)
    % \left[ \grad[\inner] \obj (\outer^{(\iter)}, \inner^{(\iter)}) + \sum_{\numconstr = 1}^\numconstrs \langmult[\numconstr]^* \grad[\inner]  \constr[\numconstr](\outer^{(\iter)}, \inner^{(\iter)}) \right] 
    \right)$ 
\EndFor
\State \Return $\{(\outer[][\iter], \inner[][\iter])\}_{\iter= 1}^\iters$
\end{algorithmic}
\end{algorithm}
 
\begin{example}
Consider the following 
this min-max Stackelberg game: $\min_{\outer[ ] \in [-1, 1]} \max_{\inner[ ] \in [-1, 1] : 1 - (\outer[ ] + \inner[ ]) \geq 0} \outer[ ]^2 - \inner[ ]^2 + 1 $.
The Stackelberg equilibrium of this game is $\outer[ ]^* = 0, \inner[ ]^* = 0$.
The Lagrangian
%associated with the outer player's value function 
is given by $\lang[{\outer[ ]}]( \inner[ ], \langmult[ ]) = \outer[ ]^2 - \inner[ ]^2 + \langmult[ ] (1 - (\outer[ ] + \inner[ ]))$.
When we plug the optimal KKT multiplier $\langmult[ ]^* = 0$ into the Lagrangian, we obtain $\lang[{\outer[ ]}]( \inner[ ], \langmult[ ]) = \outer[ ]^2 - \inner[ ]^2$.
Thus, \Cref{alg:vgda} yields the update rules $\outer[ ][\iter +1] = \outer[ ][\iter ] - 2 \outer[ ][\iter ]$ and $\inner[ ][\iter +1] = \inner[ ][\iter ] - 2\inner[ ][\iter]$.
The sequence of iterates starting at $\outer[ ][0] = 1, \inner[ ][0] = 1$, when $\learnrate[{\outer[ ]}][\iter] = \learnrate[{\inner[ ]}][\iter] = 1$ for all $\iter \in \N_+$, thus cycles as follows $\outer[ ][1] = -1, \inner[ ][1] = -1$; $\outer[ ][2] = 1, \inner[ ][2] = 1$; and so on.
Nonetheless, the average of the iterates corresponds to the Stackelberg equilibrium.
% of the game.
\end{example}
Unfortunately, LGDA
%\Cref{alg:vgda}
does not converge in general, as once $\langmult^*$ is plugged back into the Lagrangian, the Lagrangian might become degenerate in $\inner$, i.e., the dependence of the Lagrangian on $\inner$ is lost,
%\amy{what does it mean for the Lagrangian to be degenerate in $\inner$?}\deni{The $\outer$ variable might contain all the info on the $\inner$ variable so one might not be able to recover the optimal $\inner$ variable with GDA}
in which case GDA can converge to the wrong solution:
\begin{example}
Recall \Cref{ex:gda-non-convergence}.
%, and the Lagrangian associated with the outer player's value function discussed in \Cref{ex:gda-non-convergence}, namely $\lang(\outer[ ], \inner[ ], \langmult[ ]) = \outer[ ]^2 + \inner[ ] + \langmult[ ] (1 - (\outer[ ] + \inner[ ]))$.
When we plug the optimal KKT multiplier $\langmult[ ]^* = 1$ into the Lagrangian associated with the outer player's value function, we obtain $\lang[{\outer[ ]}]( \inner[ ], \langmult[ ]) = \outer[ ]^2 + \inner[ ] +  1 - (\outer[ ] + \inner[ ]) = \outer[ ]^2 - \outer[ ] +  1$, with
$\nicefrac{\partial \lang}{\partial \outer[ ]} = 2x - 1$ and $\nicefrac{\partial \lang}{\partial \inner[ ]} = 0$.
It follows that the $\outer$ iterate converges to $\nicefrac{1}{2}$, but the $\inner$ iterate will never get updated, and hence unless $\inner$ is initialized at its Stackelberg equilibirium value, LGDA will not converge to a Stackelberg equilibirium.
% $\inner[ ]^* \in [-1, 1]$, which does not necessarily correspond to a Stackelberg equilibrium.
\end{example}

This degeneracy issue arises when $\grad[\inner] \obj(\outer, \inner) = - \sum_{\numconstr = 1}^\numconstrs \langmult[\numconstr]^* \grad[\inner] \constr[\numconstr](\outer, \inner)$, $\forall \outer \in \outerset$, and can be side stepped if we restrict attention to min-max Stackelberg games with convex-\emph{strictly}-concave payoff functions, in which case the Lagrangian is guaranteed not to be degenerate in
%to be dependent on
$\inner$,
%\amy{why, if you avoid this situation ($\grad[\inner] \obj(\outer, \inner) = - \sum_{\numconstr = 1}^\numconstrs \langmult[\numconstr]^* \grad[\inner] \constr[\numconstr](\outer, \inner)$, $\forall \outer \in \outerset$) do you avoid degeneracy? what is the def'n of degeneracy?}\deni{If this condition holds the Lagrangian is guaranteed to depend on $\inner$} 
and convergence in average iterates is guaranteed.
The proof of the following theorem is relegated to \Cref{sec:app-gdadxs-proofs}:

\begin{theorem}
\label{thm:fisher-vgdad}
Let $\{(\outer[][\iter], \inner[][\iter])\}_{\iter= 1}^\iters$ be the sequence of iterates generated by LGDA
%\Cref{alg:vgda}
run on the convex-\emph{strictly}-concave min-max Stackelberg game $(\outerset, \innerset, \obj, \constr)$. Let $\left(\outer^{*}, \inner^{*}\right) \in \outerset \times \innerset$ be a Stackelberg equilibrium  of $(\outerset, \innerset, \obj, \constr)$.  Suppose that \Cref{main-assum} holds, that for all $\iter \in [\iters], \learnrate[\outer][\iter] = \learnrate[\inner][\iter] = \nicefrac{1}{\sqrt{\iters}}$. 
Let $\lipschitz[\lang] = \max_{(\outer, \inner) \in \outerset \times \innerset} \left\| \grad \lang[\outer]( \inner, \langmult^*) \right\|$, and $\avgouter[][\iters]= \frac{1}{\iters} \sum_{\iter = 1}^\iters \outer^{(\iter)}$ and $\avginner[][\iters]= \frac{1}{\iters} \sum_{\iter = 1}^\iters \inner^{(\iter)}$. We then have:
% , for all $\iter \in \N_+$:
%
\begin{align}
    -\frac{\left\|\inner[][0]-\inner^{*}\right\|^{2}+\left\|\outer[][0]-\avgouter\right\|^{2} + 2\lipschitz[\lang]^{2}}{2 \sqrt{\iters}}
    \leq \obj\left(\avgouter, \avginner\right)-\obj\left(\outer^{*}, \inner^{*}\right) \notag \\ \leq  \frac{\left\|\inner[][0]-\avginner[][\iters]\right\|^{2}+\left\|\outer[][0]-\outer^*\right\|^{2} + 2\lipschitz[\lang]^{2}}{2 \sqrt{\iters}}
\end{align}
\end{theorem}

Because of the potential degeneracy of the Lagrangian in $\inner$, we propose running the gradient ascent step for the inner player on the objective function, rather than the Lagrangian.
\Cref{alg:gdadxs}, which we dub GDALO, once again assumes access to the optimal KKT multipliers $\langmult^*$, with which it runs gradient descent on the Lagrangian for the outer player and gradient ascent on the objective function for inner player.

\begin{algorithm}[ht]
\caption{Gradient Descent Ascent with a Lagrangian 
%Solution 
Oracle}
%\sdeni{}{(GDALO)}
\label{alg:gdadxs}
\textbf{Inputs:} $\langmult^*, \outerset, \innerset, \obj, \constr, \learnrate[][\outer], \learnrate[][\inner], \iters, \outer^{(0)}, \inner^{(0)}$ \\ 
\textbf{Output:} $\outer^{*}, \inner^{*}$
\begin{algorithmic}[1]
\For{$\iter = 1, \hdots, \iters -1$}
    \State Set $\outer^{(\iter +1)} = \project[\outerset] \left( \outer^{(\iter)} - \learnrate[\iter][\outer] \lang[{\outer[][\iter]}](\inner[][\iter], \langmult^*) 
    % \left[\grad[\outer] \obj (\outer^{(\iter)}, \inner^{(\iter)}) + \sum_{\numconstr = 1}^\numconstrs \langmult[\numconstr]^* \grad[\outer]  \constr[\numconstr](\outer^{(\iter)}, \inner^{(\iter)}) \right] 
    \right)$
    
    \State Set $\inner^{(\iter +1)} = \project[\{{ \inner \in 
    \innerset : \constr (\outer^{(\iter)}, \inner) \geq 0}\}] \left( \inner^{(\iter)} + \learnrate[\iter][\inner]  \grad[\inner] \obj (\outer^{(\iter)}, \inner^{(\iter)})   \right) $
\EndFor
\State Draw $(\widehat{\outer}, \widehat{\inner})$ uniformly at random from $\{(\outer^{(\iter)}, \inner^{(\iter)})\}_{\iter =1}^\iters$
\State \Return $(\widehat{\outer}, \widehat{\inner})$
\end{algorithmic}
\end{algorithm}

%We note that, the convergence rate results as well as proofs for GDA and other known first-order methods then do not apply to this algorithm, and we have to derive the convergence rate of this algorithm from scratch. 
We are not able to prove that GDALO converges in average iterates to a Stackelberg equilibrium.
%\amy{do we have a counterexample? is it just that we cannot prove it, or that it is not true?}\deni{I was not able to prove it!}
We do show, however, that in expectation any iterate selected uniformly at random corresponds to a Stackelberg equilibrium.% 
\footnote{Note that convergence in expectation is weaker than convergence in average iterates.
%; indeed convergence in average iterates implies convergence in expectation. 
Intuitively, convergence in expectation means that as the number of iterations for which the algorithm is run increases, in expectation the output of the algorithm becomes a better and better approximation of a Stackelberg equilibrium.}
In particular, the following convergence rate holds for \Cref{alg:gdadxs}.
We refer the reader to \Cref{sec:app-gdadxs-proofs} for the proofs, which we note do not follow from known results.

\begin{theorem}
\label{thm:fisher-gdaxs}
Let $(\widehat{\outer}, \widehat{\inner})$ be the output generated by \Cref{alg:gdadxs} run on the min-max Stackelberg game $(\outerset, \innerset, \obj, \constr)$. Let $\left(\outer^{*}, \inner^{*}\right) \in \outerset \times \innerset$ be a Stackelberg equilibrium  of $(\outerset, \innerset, \obj, \constr)$.  Suppose that \Cref{main-assum} holds and that for all $\iter \in [\iters], \learnrate[\outer][\iter] = \learnrate[\inner][\iter] = \nicefrac{1}{\sqrt{\iters}}$. Let $\lipschitz[\lang] = \max_{(\outer, \inner) \in \outerset \times \innerset} \left\| \grad \lang[\outer]( \inner, \langmult^*) \right\|$ and $\lipschitz[\obj] = \max_{(\outer, \inner) \in \outerset \times \innerset} \left\| \grad \obj(\outer, \inner) \right\|$.
% , and $\avgouter[][\iters]= \frac{1}{\iters} \sum_{\iter = 1}^\iters \outer^{(\iter)}$ and $\avginner[][\iters]= \frac{1}{\iters} \sum_{\iter = 1}^\iters \inner^{(\iter)}$.
We then have:
% for all $\iter \in \N_+$,

\noindent
\begin{align}
-\frac{\left\|\inner[][0]-\inner^{*}\right\|^{2} + \lipschitz[\obj]^{2}}{2 \sqrt{\iters}} \leq \Ex \left[\obj\left((\widehat{\outer}, \widehat{\inner})\right) \right] - \obj\left(\outer^{*}, \inner^{*}\right)  \leq \frac{\left\|\outer[][0]-\outer^{*}\right\|^{2} + \lipschitz[\lang]^{2}}{2  \sqrt{\iters}}
\end{align}

% (b) 
% $$
% -\frac{\left\|\inner[][0]-\inner^{*}\right\|^{2}+\left\|\outer[][0]-\avgouter\right\|^{2}}{2 \learnrate[ ] \iters}-\learnrate[ ] L^{2} \leq \obj\left(\avgouter, \avginner\right)-\obj\left(\outer^{*}, \inner^{*}\right) \leq \frac{\left\|\outer[][0]-\outer^{*}\right\|^{2}+\left\|\inner[][0]-\avginner\right\|^{2}}{2 \learnrate[ ] \iters}+\learnrate[ ] L^{2},
% $$

\end{theorem}

We conclude this section by providing an explicit closed-form solution for the optimal KKT multipliers for a large class of min-max Stackelberg games.
The existence of a Lagrangian oracle is guaranteed in these games.

\begin{theorem}
\label{thm:langrangian-oracle-existence}
    Consider a min-max Stackelberg game of the following form:
    \begin{align}
        \min_{\outer \in \outerset} \max_{\Y \in \innerset^\outerdim: \forall i, \constr[i](\y_i,\outer) \leq c_i } \obj_1(\outer) + \sum_{i = 1}^\outerdim a_i \log(\obj_2(\outer, \y_i))  + \sum_{i = 1}^\outerdim b_i \log(\obj_3(\y_i))
    \end{align}
    where $\obj_1: \outerset \to \R$, $\obj_2: \innerset \to \R$, $\obj_3: \innerset \to \R$, $\constr: \outerset \times \innerset \to \R$ and $\outerset \subset \R^\outerdim, \innerset \subset \R^{\innerdim}$ are compact-convex.
    Suppose that 
    1.~ $\obj_2$, $\obj_3$, $\constr[1], \hdots, \constr[\outerdim]$ are concave $\inner$, for all $\outer \in \outerset$,
    2.~ $\obj_2$, $\obj_3$ are homogeneous in $\inner$, for all $\outer \in \outerset$, i.e., $\forall k \in \R f_2(\outer, k \inner) = k f_2(\outer, \inner), f_3(k \inner) = k f_3(\inner)$,
    and continuous.
    Then, the optimal KKT multipliers $\langmult^* \in \R^\outerdim$ are $\langmult[i]^* = \frac{a_i + b_i }{c_i}$, for all $i \in [\outerdim]$.
\end{theorem}
We remark that since the optimal KKT multipliers of the min-max Stackelberg games of the form of \Cref{thm:langrangian-oracle-existence} are given in closed form, such games can be converted into min-max games with independent strategy sets.
That said, \Cref{alg:gdadxs} is still necessary, because as we have shown, the Lagrangian can become degenerate in which case LGDA can converge to the wrong solution.

%% file: experiments.tex
\section{Application to
%An Economic Application: 
Fisher Markets}
\label{sec:fisher}

The Fisher market model, attributed to Irving Fisher \cite{brainard2000compute}, has received a great deal of attention recently, in particular by computer scientists, as its applications to fair division and mechanism design have proven useful for the design of automated markets in many online marketplaces.
In this section, we use our algorithms to compute competitive equilibria in Fisher markets, which have been shown to be instances of min-max Stackelberg games \cite{goktas2021minmax}.

\if 0
so that their Stackelberg equilibria coincide with the competitive equilibria.
Applied to the computation of competitive equilibria in Fisher markets, our GDA algorithms correspond to myopic best-response dynamics.
In a related dynamic price-adjustment process called t\^atonnement \cite{walras}, which also converges to competitive equilibria, sellers adjust their prices incrementally, while buyers respond optimally to the sellers' price adjustments.
Our dynamics give rise to a novel t\^atonnement process in which both buyers and sellers exhibit bounded rationality, with sellers adjusting prices incrementally as usual, and buyers adjusting their demands incrementally as well.
\fi

%\subsection{Model}

A \mydef{Fisher market} consists of $\numbuyers$ buyers and $\numgoods$ divisible goods \cite{brainard2000compute}.
Each buyer $\buyer \in \buyers$ has a budget $\budget[\buyer] \in \mathbb{R}_{++}$ 
%\amy{two plusses or one; there is only one plus in the arxiv.}\deni{In this specific paper I assume for all $\buyer$ that $\budget[\buyer] > 0$ because it allows for simpler proofs. In particular, otherwise I would have to deal with a division by 0. This assumption is without loss of generality because we can remove any buyer with 0 budget from the market without changing equilibria.} 
and a utility function $\util[\buyer]: \mathbb{R}_{+}^{\numgoods} \to \mathbb{R}$.
%\footnote{\sdeni{WLOG, we assume that all budgets are strictly positive; otherwise, we could remove a buyer with a zero budget without affecting the equilibrium.}{}} 
As is standard in the literature, we assume that there is one divisible unit of each good available in the market \cite{AGT-book}.
An instance of a Fisher market is given by a tuple $(\numbuyers, \numgoods, \util, \budget)$, where $\util = \left\{\util[1], \hdots, \util[\numbuyers] \right\}$ is a set of utility functions, one per buyer, and $\budget \in \R_{+}^{\numbuyers}$ is the vector of buyer budgets.
We abbreviate as $(\util, \budget)$ when $\numbuyers$ and $\numgoods$ are clear from context.

An \mydef{allocation} $\allocation = \left(\allocation[1], \hdots, \allocation[\numbuyers] \right)^T \in \R_+^{\numbuyers \times \numgoods}$ is a map from goods to buyers, represented as a matrix, s.t.\ $\allocation[\buyer][\good] \ge 0$ denotes the amount of good $\good \in \goods$ allocated to buyer $\buyer \in \buyers$.
Goods are assigned \mydef{prices} $\price = \left(\price[1], \hdots, \price[\numgoods] \right)^T \in \mathbb{R}_+^{\numgoods}$.
A tuple $(\price^*, \allocation^*)$ is said to be a \mydef{competitive (or Walrasian) equilibrium} of Fisher market $(\util, \budget)$ if 1.~buyers are utility maximizing, constrained by their budget, i.e., $\forall \buyer \in \buyers, \allocation[\buyer]^* \in \argmax_{\allocation[ ] : \allocation[ ] \cdot \price^* \leq \budget[\buyer]} \util[\buyer](\allocation[ ])$;
and 2.~the market clears, i.e., $\forall \good \in \goods,  \price[\good]^* > 0 \Rightarrow \sum_{\buyer \in \buyers} \allocation[\buyer][\good]^* = 1$ and $\price[\good]^* = 0 \Rightarrow\sum_{\buyer \in \buyers} \allocation[\buyer][\good]^* \leq 1$.

\citeauthor{goktas2021minmax} \cite{goktas2021minmax} observe that any competitive equilibrium $(\price^*, \allocation^*)$ of a Fisher market $(\util, \budget)$ corresponds to a Stackelberg equilibrium of the following min-max Stackelberg game:
\begin{align}
    \min_{\price \in \R_+^\numgoods} \max_{\allocation \in \R^{\numbuyers \times \numgoods}_+ :  \allocation \price \leq \budget} \sum_{\good \in \goods} \price[\good] + \sum_{\buyer \in \buyers}  \budget[\buyer] \log \left(  \util[\buyer](\allocation[\buyer]) \right) \enspace .
    \label{fisher-program}
\end{align}
\noindent
By \Cref{thm:langrangian-oracle-existence}, we can obtain a Lagrangian solution oracle for \Cref{fisher-program}, given by the following corollary. We note that as for all buyers $\buyer \in \buyers$, $\budget[\buyer] > 0$, Slater's condition is satisfied.
\begin{corollary}
\label{corr:fisher-lagrangian-oracle}
Consider the min-max Stackelberg game described by \Cref{fisher-program}. The optimal KKT multipliers are given by  %$\langmult^* \in \R^\numbuyers$ are given by
$\langmult^* = \ones[\numbuyers]$.
\end{corollary}
Let $\lang$ be the Lagrangian of the outer player's value function in \Cref{fisher-program}, i.e., $\lang[\price](\allocation, \langmult) = \sum_{\good \in \goods} \price[\good] + \sum_{\buyer \in \buyers}  \budget[\buyer] \log \left(  \util[\buyer](\allocation[\buyer]) \right) + \sum_{\buyer \in \buyers} \langmult[\buyer] \left( \budget[\buyer] - \allocation[\buyer] \cdot \price \right)$.
Using \Cref{corr:fisher-lagrangian-oracle}, we can define 
\mydef{myopic best-response dynamics} (\Cref{alg:myopic-br}; MBRD) \cite{cournot1897researches, monderer1996potential} in Fisher markets as GDALO run on \Cref{fisher-program}, by noting that $\grad[\price] \lang[\price](\allocation, \ones[\numbuyers]) = \\ \ones[\numgoods] - \sum_{\buyer \in \buyers} \allocation[\buyer]$
%This statement follows from
(\citeauthor{goktas2022cch}  \cite{goktas2022cch}, Theorem 3).

\begin{algorithm}[ht]
\caption{Myopic Best-Response Dynamics (MBRD)}
\label{alg:myopic-br}
\textbf{Inputs:} $\util, \budget, \learnrate[][\price], \learnrate[][\allocation], \iters, \allocation^{(0)}, \price^{(0)}$ \\ 
\textbf{Output:} $\outer^{*}, \inner^{*}$
\begin{algorithmic}[1]
\For{$\iter = 1, \hdots, \iters -1$}
    \State Set $\price^{(\iter +1)} = \project[\R_+^\numgoods] \left( \price^{(\iter +1)} + \sum_{\buyer \in \buyers} \allocation[\buyer]^{(t)} - \ones[\numgoods] \right)$
    \State For all $\buyer \in \buyers$, set \\ 
    $\allocation[\buyer]^{(\iter +1)} = \project[\{{ \allocation[ ] \in 
    \R^\numgoods_+ : \allocation[ ] \cdot \price^{(\iter-1)} \leq \budget[\buyer]}\}] \left( \allocation[\buyer]^{(\iter)} + \learnrate[\iter][\inner]  \frac{\budget[\buyer]}{\util[\buyer]\left(\allocation[\buyer]^{(\iter)}\right)} \grad[{\allocation[\buyer]}] \util[\buyer]\left(\allocation[\buyer]^{(\iter)}\right) \right)$
\EndFor
\State Draw $(\widehat{\price}, \widehat{\allocation})$ uniformly at random from $\{(\price^{(\iter)}, \allocation^{(\iter)})\}_{\iter = 1}^\iters$
\State \Return $(\widehat{\price}, \widehat{\allocation})$ 
\end{algorithmic}
\end{algorithm}

In words, under myopic best-response dynamics, at each time step the (fictional Walrasian) auctioneer takes a gradient descent step, and then all the buyers take a gradient ascent step to maximize their utility.
We note that myopic best-response dynamics can also be interpreted as a t\^atonnement process run with boundedly rational buyers who take a step in the direction of their optimal bundle, but do not actually compute their optimal bundle at each time step.
We thus have the following
%(informal)
corollary of \Cref{thm:fisher-gdaxs}:

\begin{corollary}
Let $(\util, \budget)$ be a Fisher market with equilibrium price vector $\price^{*}$, where $\util$ is a set of continuous, concave, homogeneous, and continuously differentiable utility functions.
Running myopic best-response dynamics (\Cref{alg:myopic-br}) on the Fisher market $(\util, \budget)$ yields an output which is in expectation an $\varepsilon$-competitive equilibrium with $\varepsilon$-utility maximizing allocations in $O(\nicefrac{1}{\varepsilon^2})$ iterations.%
\footnote{We note that one can ensure that the derivatives of $\sum_{\buyer \in \buyers} \budget[\buyer] \log(\util[\buyer](\allocation[\buyer]))$ w.r.t. $\allocation[\buyer][\good]$ in \Cref{fisher-program} are bounded at $\allocation[\buyer][\good] = 0$ for all $\buyer \in \buyers$ and $\good \in \goods$ by reparametrizing the program as $\sum_{\good \in \goods} \price[\good] + \sum_{\buyer \in \buyers} \log(\util[\buyer](\allocation[\buyer]) + \delta)$ for $\delta > 0$.}
%\amy{i don't understand this! want to understand it.}\deni{At a high-level I change the objective so that it always has a bounded derivative but this means that I am also shifting the optimal value of the SE by a $\delta$ factor, so this means I need a constant number moreof iterations to get into a $\delta$ eq. but this does not break the theorem statement, because everything still holds in big-oh notation!}
\end{corollary}

\paragraph{Experiments}
As we are not able to prove average-iterate convergence for GDALO, we ran a series of experiments on Fisher markets 
%with three different classes of utility functions 
in which we track whether or not the sequence of average iterates produced by MBRD converges to a competitive equilibrium.%
\footnote{Our code can be found at \coderepo.}
We consider three buyer utility structures, each of which endows \Cref{fisher-program} with different smoothness and convexity properties, thus allowing us to compare the efficiency of the algorithms under these varying conditions.
\if 0
Let $\valuation[\buyer] \in \R^\numgoods$ be a vector of parameters that describes the utility function of buyer $\buyer \in \buyers$.
We consider the following utility function classes:
linear: $\util[\buyer](\allocation[\buyer]) = \sum_{\good \in \goods} \valuation[\buyer][\good] \allocation[\buyer][\good]$,  Cobb-Douglas:  $\util[\buyer](\allocation[\buyer]) = \prod_{\good \in \goods} \allocation[\buyer][\good]^{\valuation[\buyer][\good]}$, Leontief:  $\util[\buyer](\allocation[\buyer]) = \min_{\good \in \goods} \left\{ \frac{\allocation[\buyer][\good]}{\valuation[\buyer][\good]}\right\}$.
% \amy{is $V$ defined? $V$ being differentiable was important in the last paper. is it also important here?}
For linear and Cobb-Douglas Fisher markets, \Cref{main-assum} is satisfied.
Additionally, for Cobb-Douglas and Leontief Fisher markets, the objective function is linear-strictly-concave, while it is linear-concave for linear Fisher markets.
\fi 
We summarize the properties of \Cref{fisher-program} in these three Fisher markets in Table~\ref{tab:util-prop}, and include a more detailed description of our experimental setup in \Cref{sec-app:fisher}.

\begin{table}[ht]
\caption{\label{tab:util-prop}
Smoothness and convexity properties 
%satisfied by \Cref{fisher-program} 
assuming different utility functions. Note that \Cref{main-assum} does not hold for Leontief utilities, because they are not differentiable.}
\begin{center}
\begin{tabular}{|p{0.08\textwidth}|p{0.06\textwidth}|p{0.06\textwidth}|p{0.1\textwidth}|}\hline
 Buyer utilities & Linear-concave & Linear-strictly concave & \Cref{main-assum} holds \\ \hline
Linear & \checkmark & \checkmark & \checkmark\\ \hline
Cobb-Douglas & $\times$  & \checkmark & \checkmark\\ \hline
Leontief & $\times$ & $\times$  & \checkmark \\\hline
\end{tabular}
\end{center}
\end{table}  

%Although not guaranteed by our theory, in all our experiments with MBRD, we observe convergence in average iterates to a competitive equilibrium, albeit at varying rates. 
Let $\{\outer[][\iter], \inner[][\iter]\}_{\iter =1}^\iters $ be the sequence of iterates generated by MBRD and let $\bar{\price}^\iter = \nicefrac{1}{\iter} \sum_{\iter = 1}^\iters \price^{(\iter)}$.
\Cref{fig:experiments} then depicts average exploitability,
%at each average iterate, 
i.e. $\forall \iter \in \N_+$,
$$\max_{\allocation \in \R^{\numbuyers \times \numgoods}_+ : \allocation \bar{\price}^\iter \leq \budget} \sum_{\good \in \goods} \bar{\price}_{\good}^\iter + \sum_{\buyer \in \buyers}  \budget[\buyer] \log \left(  \util[\buyer](\allocation[\buyer]) \right) -$$
$$\min_{\price \in \R_+^\numgoods} \max_{\allocation \in \R^{\numbuyers \times \numgoods}_+ :  \allocation \price \leq \budget} \sum_{\good \in \goods} \price[\good] + \sum_{\buyer \in \buyers}  \budget[\buyer] \log \left(  \util[\buyer] (\allocation[\buyer]) \right) \enspace ,$$
across all runs, divided by $\nicefrac{1}{\sqrt{T}}$. Note that if this quantity is constant then the (average) iterates converge empirically, at a rate of $O(\nicefrac{1}{\sqrt{\iters}})$, while if it is an increasing (decreasing) function then the iterates converge empirically at a rate slower (faster) than $O(\nicefrac{1}{\sqrt{\iters}})$.
%for linear, Cobb-Douglas, and Leontief utilities.

Convergence is fastest in Fisher markets with Cobb-Douglas utilities, followed by linear, and then Leontief.
Linear and Cobb-Douglas Fisher markets appears to converge at a rate faster than $O(\nicefrac{1}{\sqrt{T}})$ in both 
%, although convergence seems faster for Cobb-Douglas markets. 
% (i.e., if the dashed red-line differs from the blue line only by a constant factor), or if the convergence is slightly slower than $O(\nicefrac{1}{\sqrt{T}})$. 
Like for linear utilities, the objective function is twice differentiable for Cobb-Douglas utilities, but for Cobb-Douglas utilities the objective is also linear-strictly-concave, which could explain the faster convergence rate.
Fisher markets with Leontief utilities, in which the objective function is not differentiable, are the hardest markets of the three for our algorithms to solve; we seem to obtain a convergence rate slower than $O(\nicefrac{1}{\sqrt{T}})$.
Our experiments suggest that the convergence of MBRD  (\Cref{alg:myopic-br}), and more generally, GDALO (\Cref{alg:gdadxs}) can be improved to convergence in average iterates at a $O(\nicefrac{1}{\sqrt{T}})$ rate, when \Cref{main-assum} holds.
% Finally, it is also plausible that this convergence rate could also be extended to all convex-concave objective functions since a $O(\nicefrac{1}{\sqrt{T}})$ also seems plausible for linear Fisher markets.

\begin{figure}[h!]
    \centering
        \caption{Average exploitability 
        %for the outer player evaluated at the average iterates 
        divided by $\nicefrac{1}{\sqrt{\iters}}$ after running MBRD (\Cref{alg:myopic-br})  on randomly initialized linear, Cobb-Douglas, and Leontief Fisher markets, respectively.}
    % In {\color{red} red}, an arbitrary $O(\nicefrac{1}{\sqrt{T}})$ function is shown. In {\color{ForestGreen} green}, an arbitrary $O(\nicefrac{1}{T^{\nicefrac{1}{3}}})$ function is shown.}
    \label{fig:experiments}
    \includegraphics[scale=0.20]{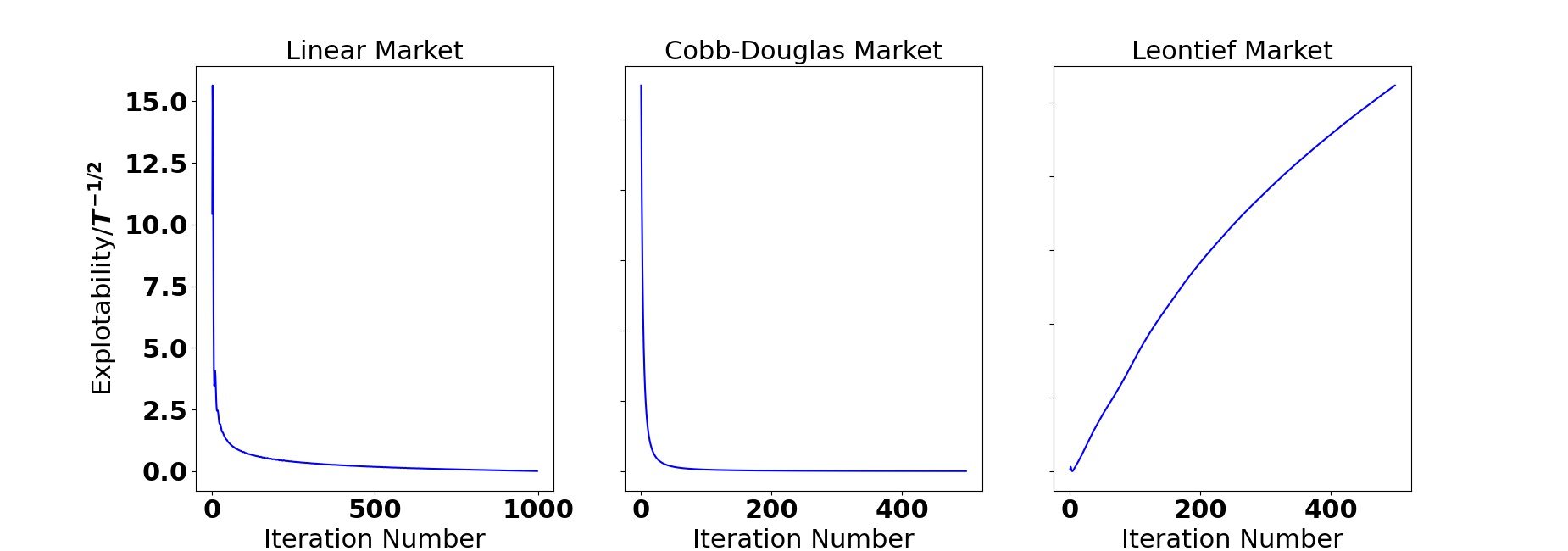}
\end{figure}

%% file: conclusion.tex
\section{Conclusion}
\label{sec:conc}

We began this paper by observing that a straightforward generalization of GDA to min-max Stackelberg games does not converge to a Stackelberg equilibrium in convex-concave min-max Stackelberg games.
We then introduced two variants of GDA that do converge in polynomial time to Stackelberg equilibria in a large class of min-max Stackelberg games.
Both of our algorithms, LGDA and GDALO, converge in $O(\nicefrac{1}{\varepsilon^2})$ iterations to an $\varepsilon$-Stackelberg equilibrium.
While LGDA converges in averages iterates only in convex-\emph{strictly}-concave min-max Stackelberg games under standard smoothness assumptions, GDALO converges in expected iterates in all convex-concave min-max Stackelberg games.
The iteration complexity of these algorithms improve on the state-of-the-art nested GDA algorithm proposed by \citeauthor{goktas2021minmax}   \cite{goktas2021minmax}, which computes an equilibrium in $O(\nicefrac{1}{\varepsilon^3})$ iterations.

We then applied GDALO the computation of competitive equilibria in Fisher markets, which yielded myopic best-response dynamics for Fisher markets---a new form of t\^atonnement in which both the buyers and sellers are boundedly rational.
Our experiments suggest that GDALO's expected iterate convergence can be improved to average 
% \amy{last or average?}\deni{Average}\amy{so fix please!}
iterate convergence, since in all experiments, the average iterates do indeed converge to competitive equilibria.
Additionally, our experiments suggest avenues for future work, namely investigating how varying the degree of smoothness of the objective function can impact convergence rates.

%% file: appendix/fisher.tex
\section{An Economic Application: Details}\label{sec-app:fisher}

Our experimental goal was to understand if GDALO (\Cref{alg:gdadxs}) converges in average iterates and if so how the rate of convergences changes under different utility structures, i.e. different smoothness and convexity properties of the objective function. 

To achieve this goal, we ran multiple experiments, each time recording the prices and allocations during each iteration $t$ of the main (outer) loop. We checked for each experiment if the average iterates converged. We have found no experiment for which there was no average iterate convergence.
% For each run of the algorithm on each market with each set of initial conditions, we then computed the outer player's value function under the average iterate strategy which we plot in \Cref{fig:experiments}. \amy{this is NOT what you plot}
Let $\{\outer[][\iter], \inner[][\iter]\}_{\iter =1}^\iters $ be the sequence of iterates generated by MBRD and let $\bar{\price}^\iter = \nicefrac{1}{\iter} \sum_{\iter = 1}^\iters \price^{(\iter)}$.
\Cref{fig:experiments} then depicts average exploitability,
%at each average iterate, 
i.e. $\forall \iter \in \N_+$,
$$\max_{\allocation \in \R^{\numbuyers \times \numgoods}_+ : \allocation \bar{\price}^\iter \leq \budget} \sum_{\good \in \goods} \bar{\price}_{\good}^\iter + \sum_{\buyer \in \buyers}  \budget[\buyer] \log \left(  \util[\buyer](\allocation[\buyer]) \right) -$$
$$\min_{\price \in \R_+^\numgoods} \max_{\allocation \in \R^{\numbuyers \times \numgoods}_+ :  \allocation \price \leq \budget} \sum_{\good \in \goods} \price[\good] + \sum_{\buyer \in \buyers}  \budget[\buyer] \log \left(  \util[\buyer] (\allocation[\buyer]) \right) \enspace ,$$
across all runs, divided by $\nicefrac{1}{\sqrt{T}}$.

\paragraph{Hyperparameters}
We randomly initialized 500 different linear, Cobb-Douglas, Leontief Fisher markets, each with $5$ buyers and $8$ goods.
Buyer $\buyer$'s budget $\budget[\buyer]$ was drawn randomly from a uniform distribution ranging from $10$ to $20$ (i.e., $U[10,20]$), while each buyer $\buyer$'s valuation for good $\good$, $\valuation[i][j]$, was drawn randomly from $U[5,15]$.
We ran MBRD (\Cref{alg:myopic-br}) for 1000, 500, and 500 iterations for linear, Cobb-Douglas, and Leontief Fisher markets, respectively.
We started the algorithm with initial prices drawn randomly from $U[5,15]$.
%, to understand the impact of the initial conditions on the algorithms' behaviors.
%
After manual hyper-parameter tuning, we opted for fixed learning rates of $\forall \iter \in \N_+, \learnrate[\outer][\iter] = 3$, $\learnrate[\inner][\iter] = 1$ for Cobb-Douglas and Leontief Fisher markets, while we picked fixed learning rates of $\forall \iter \in \N_+, \learnrate[\outer][\iter] = 3$, $\learnrate[\inner][\iter] = 0.1$ for linear Fisher markets.

\paragraph{Programming Languages, Packages, and Licensing}
We ran our experiments in Python 3.7 \cite{van1995python}, using NumPy \cite{numpy}, Pandas \cite{pandas}, and CVXPY \cite{diamond2016cvxpy}.
\Cref{fig:experiments} was graphed using Matplotlib \cite{matplotlib}.

Python software and documentation are licensed under the PSF License Agreement. Numpy is distributed under a liberal BSD license. Pandas is distributed under a new BSD license. Matplotlib only uses BSD compatible code, and its license is based on the PSF license. CVXPY is licensed under an APACHE license.

\paragraph{Implementation Details}
In order to project each allocation computed onto the budget set of the consumers, i.e., $\{\allocation \in \R^{\numbuyers \times \numgoods}_+ \mid \allocation\price \leq \budget\}$, we used the alternating projection algorithm for convex sets, and alternatively projected onto the sets $\R^{\numbuyers \times \numgoods}_+$ and $\{\allocation \in \R^{\numbuyers \times \numgoods} \mid \allocation\price \leq \budget\}$.

\paragraph{Computational Resources}
Our experiments were run on MacOS machine with 8GB RAM and an Apple M1 chip, and took about 2 hours to run. Only CPU resources were used.

\paragraph{Code Repository}
The data our experiments generated, as well as the code used to produce our visualizations and run the statistical tests, can be found in our code repository \\({\color{blue}\rawcoderepo}).

%% file: appendix/proofs/simGDA_proofs.tex
\section{Omitted Proofs Section \ref{sec:GDA}}\label{sec:app-GDA-proofs}

\begin{proof}[Proof of \cref{thm:stackelberg-equiv}]
By the definition of the Lagrangian,\\ $\max_{\inner \in \innerset : \constr(\outer, \inner) \geq \zeros} \obj(\outer, \inner) =
\max_{\inner \in \innerset } \inf_{\langmult \geq \zeros} \lang[\outer]( \inner, \langmult)$.
Slater's condition guarantees that the infimum exists, so that it can be replaced by a minimum.
Thus, $\max_{\inner \in \innerset : \constr(\outer, \inner) \geq \zeros} \obj(\outer, \inner) =
\max_{\inner \in \innerset } \min_{\langmult \geq \zeros} \lang[\outer]( \inner, \langmult)$.
In other words,
%under \Cref{main-assum},
we can re-express the inner player's maximization problem as a convex-concave min-max game via the Lagrangian.
Slater's condition also implies strong duality, meaning $\max_{\inner \in \innerset } \min_{\langmult \geq \zeros} \lang[\outer]( \inner, \langmult) = \min_{\langmult \geq \zeros} \max_{\inner \in \innerset } \lang[\outer]( \inner, \langmult)$.
It follows that we can re-express any min-max Stackelberg game with dependent strategy sets as a three-player game with payoff function $\lang[\outer](\inner, \langmult)$, where the outer player, $\outer$, moves first and the inner two players, $\inner$ and $\langmult$, play a Nash equilibrium.
\if 0
: i.e.,
%
\begin{align}
    \min_{\outer \in \outerset} \max_{\inner \in \innerset : \constr(\outer, \inner) \geq \zeros} \obj(\outer, \inner) = \min_{\outer \in \outerset} \max_{\inner \in \innerset } \min_{\langmult \geq \zeros} \lang[\outer]( \inner, \langmult)
    \enspace ,
\end{align} 
where the order of play of the $\langmult$ and $\inner$ players can be reversed.
%, while the order of play of the $\outer$ player cannot be changed.
\fi
\end{proof}

%% file: appendix/proofs/gdaxs_proofs.tex
\section{Omitted Proofs Section \ref{sec:oracle}}\label{sec:app-gdadxs-proofs}

\begin{proof}[Proof Sketch of \Cref{thm:fisher-vgdad}]
Note that when $\obj$ is convex-strictly-concave in $\outer$ and $\inner$, then $\lang[\outer](\inner, \langmult^*) = \obj(\outer, \inner) + \sum_{\numconstr = 1}^\numconstrs \langmult[\numconstr]^* \constr[\numconstr](\outer, \inner)$ is also convex-strictly-concave, since $\sum_{\numconstr = 1}^\numconstrs \langmult[\numconstr]^* \constr[\numconstr](\outer, \inner)$ is convex-concave, and the sum of a convex-concave function and a convex-strictly-concave function is convex-strictly-concave.
This implies that for all $\outer \in \outerset$, there exists a unique $\inner^*(\outer)$ that solves $\lang[\outer](\inner, \langmult^*)$.
Thus, the Lagrangian is well-defined, i.e., non-null, and we can recover the optimal primal variables from the Lagrangian \cite{boyd2004convex}.
Convergence of GDA follows as a direct extension of \citeauthor{nedic2009gda}'s Theorem 3.1 \cite{nedic2009gda}.
\end{proof}

\begin{lemma}\label{lemma:basic-iterate-relationships}
Let the sequences $\left\{\outer[][\iter]\right\}$ and $\left\{\inner[][\iter]\right\}$ be generated by \Cref{alg:gdadxs}. Suppose that \cref{main-assum} holds. Let $\lipschitz[\lang] = \max_{(\outer, \inner) \in \outerset \times \innerset}\\ \left\| \grad \lang[\outer]( \inner, \langmult^*) \right\|$ and $\lipschitz[\obj] = \max_{(\outer, \inner) \in \outerset \times \innerset} \left\| \grad \obj(\outer, \inner) \right\|$, we then have:

\noindent
(a) For any $\outer \in \outerset$ and for all $\iter \in \N_+$,
\begin{align}
&\left\|\outer[][\iter+1]-\outer\right\|^{2} \leq\left\|\outer[][\iter]-\outer\right\|^{2}-2 \learnrate[ ]\left(\lang[{\outer[][\iter]}]\left( \inner[][\iter], \langmult^* \right)-\lang[{\outer}]\left( \inner[][\iter], \langmult^*\right)\right)+{\learnrate[ ]}^{2}\lipschitz[\lang]^{2} .
\end{align}
(b) For any $\inner \in \innerset$ and for all $\iter \in \N_+$,
\begin{align}
&\left\|\inner[][\iter + 1]-\inner\right\|^{2}  \leq\left\|\inner[][\iter]-\inner\right\|^{2}+2 \learnrate[ ]\left(\obj\left(\outer[][\iter ], \inner[][\iter]\right)-\obj\left(\outer[][\iter ], \inner\right)\right)+{\learnrate[ ]}^{2}\lipschitz[\obj]^{2}
\end{align}
\end{lemma} 

\begin{proof}
(a) By the non-expansivity of the projection operation and the definition of \cref{alg:gdadxs} we obtain for any $\outer \in \outerset$ and all $\iter \in \N_+$,

\begin{align}
\left\|\outer[][\iter+1]-\outer\right\|^{2} &=\left\|\project[\outerset]\left[\outer[][\iter]-\learnrate[ ] \grad[\outer]\lang[{\outer[][\iter]}]\left(\inner[][\iter], \langmult^* \right)\right]-\outer\right\|^{2} \\
& \leq\left\|\outer[][\iter]-\learnrate[ ] \grad[\outer]\lang[{\outer[][\iter]}]\left( \inner[][\iter], \langmult^* \right) - \outer\right\|^{2} \\
&=\left\|\outer[][\iter]-\outer\right\|^{2}  -2 \learnrate[ ] \grad[\outer]\lang[{\outer[][\iter]}]\left( \inner[][\iter], \langmult^* \right)^T\left(\outer[][\iter]-\outer\right)  + {\learnrate[ ]}^{2}\left\|\grad[\outer]\lang[{\outer[][\iter]}]\left( \inner[][\iter], \langmult^* \right)\right\|^{2} .
\end{align}

Since the Lagrangian function $\lang[\outer](\inner, \langmult^*)$ is convex in $\outer$ for each $\inner \in \innerset$, and since $\grad[\outer]\obj\left(\outer[][\iter], \inner[][\iter]\right)$ is a subgradient of $\obj\left(\outer, \inner[][\iter]\right)$ with respect to $\outer$ at $\outer=\outer[][\iter]$, we obtain for any $\outer \in \outerset$,
\begin{align}
&\grad[\outer]\lang[{\outer[][\iter]}]\left( \inner[][\iter], \langmult^* \right)^T\left(\outer-\outer[][\iter]\right) \leq \lang[{\outer}]\left( \inner[][\iter], \langmult^*\right) - \lang[{\outer[][\iter]}]\left( \inner[][\iter], \langmult^*\right)
\end{align}
or equivalently
\begin{align}
-\grad[\outer]\lang[{\outer[][\iter]}]\left( \inner[][\iter], \langmult^* \right)^T\left(\outer[][\iter] - \outer\right)   \leq - \left(\lang[{\outer[][\iter]}]\left( \inner[][\iter], \langmult^*\right) - \lang[{\outer}]\left( \inner[][\iter], \langmult^*\right)  \right)
\end{align}
Hence, for any $\outer \in \outerset$ and all $\iter \in \N_+$,
\begin{align}
\left\|\outer[][\iter+1]-\outer\right\|^{2} & \leq\left\|\outer[][\iter]-\outer\right\|^{2}   
-2 \learnrate[ ]\left(\lang[{\outer[][\iter]}]\left( \inner[][\iter], \langmult^* \right)-\lang[{\outer}]\left( \inner[][\iter], \langmult^*\right)\right) +{\learnrate[ ]}^{2}\left\|\grad[\outer]\lang[{\outer[][\iter]}]\left( \inner[][\iter], \langmult^*\right)\right\|^{2} .
\end{align}
Since $\lang$ is continuously differentiable, it is $\lipschitz[\lang]$-Lipschitz continuous with $\lipschitz[\lang] = \max_{(\outer, \inner) \in \outerset \times \innerset} \left\| \grad \lang[\outer]( \inner, \langmult^*) \right\|$. Hence, we have for any $\outer \in \outerset$ and $\iter \in \N_+$:
\begin{align}
&\left\|\outer[][\iter+1]-\outer\right\|^{2} \leq\left\|\outer[][\iter]-\outer\right\|^{2}  -2 \learnrate[ ]\left(\lang[{\outer[][\iter]}]\left( \inner[][\iter], \langmult^* \right)-\lang[{\outer}]\left( \inner[][\iter], \langmult^*\right)\right) \\
&+{\learnrate[ ]}^{2}\lipschitz[\lang]^{2} .
\end{align}
(b) Similarly, by using the non-expansivity of the projection operation and the definition of \cref{alg:gdadxs} we obtain for any $\inner \in \innerset$ and for all $\iter \in \N_+$:
\begin{align}
\left\|\inner[][\iter + 1]-\inner\right\|^{2} 
&= \left\|\project[\{{ \inner \in 
    \innerset : \constr (\outer^{(\iter)}, \inner) \geq 0}\}] \left[ \inner^{(\iter)} + \right. \right. \left. \left. \learnrate[ ][ ]  \grad[\inner] \obj (\outer^{(\iter)}, \inner^{(\iter)})   \right] - \inner \right\|^2\\
&\leq \left\| \inner^{(\iter)}   +  \learnrate[ ][ ]  \grad[\inner] \obj (\outer^{(\iter)}, \inner^{(\iter)})   - \inner \right\|^2\\
&\leq \left\|\inner[][\iter]-\inner\right\|^{2}  + 2 \learnrate[ ]\left(\inner[][\iter]-\inner\right)^T \grad[\inner]\obj\left(\outer[][\iter ], \inner[][\iter]\right) +{\learnrate[ ]}^{2}\left\|\grad[\inner]\obj\left(\outer[][\iter ], \inner[][\iter]\right)\right\|^{2} 
\end{align}

\noindent
Since $\grad[\inner]\obj\left(\outer[][\iter ], \inner[][\iter]\right)$ is a subgradient of the concave function $\obj\left(\outer[][\iter ], \inner\right)$ at $\inner=\inner[][\iter]$, we have for all $\inner \in \innerset$,
\begin{align}
\grad[\inner]\obj\left(\outer[][\iter ], \inner[][\iter]\right)^T \left(\inner[][\iter]-\inner\right)  
\leq \obj\left(\outer[][\iter ], \inner[][\iter]\right)-\obj\left(\outer[][\iter ], \inner\right)
\end{align}
Hence, for any $\inner \in \innerset$ and all $\iter \in \N_+$
\begin{align}
&\left\|\inner[][\iter + 1]-\inner\right\|^{2} \leq\left\|\inner[][\iter]-\inner\right\|^{2} +2 \learnrate[ ]\left(\obj\left(\outer[][\iter ], \inner[][\iter]\right)-\obj\left(\outer[][\iter ], \inner\right)\right)  +{\learnrate[ ]}^{2}\left\|\grad[\inner]\obj\left(\outer[][\iter ], \inner[][\iter]\right)\right\|^{2}
\end{align}
Since $\obj$ is continuously differentiable, it is $\lipschitz[\obj]$-Lipschitz continuous with $\lipschitz[\obj] = \max_{(\outer, \inner) \in \outerset \times \innerset} \left\| \grad \obj(\outer, \inner) \right\|$. Hence, we have for any $\inner \in \innerset$ and $\iter \in \N_+$:
\begin{align}
&\left\|\inner[][\iter + 1]-\inner\right\|^{2} \leq\left\|\inner[][\iter]-\inner\right\|^{2} + 2 \learnrate[ ]\left(\obj\left(\outer[][\iter ], \inner[][\iter]\right)-\obj\left(\outer[][\iter ], \inner\right)\right)+{\learnrate[ ]}^{2}\lipschitz[\obj]^{2}
\end{align}
\end{proof}

\begin{lemma}\label{lemma:GDA-main-lemma}
Let the sequences $\left\{\outer[][\iter ]\right\}$ and $\left\{\inner[][\iter]\right\}$ be generated by \Cref{alg:gdadxs}. Suppose that \cref{main-assum} holds. Let $\lipschitz[\lang] = \max_{(\outer, \inner) \in \outerset \times \innerset}\\ \left\| \grad \lang[\outer]( \inner, \langmult^*) \right\|$ and $\lipschitz[\obj] = \max_{(\outer, \inner) \in \outerset \times \innerset} \left\| \grad \obj(\outer, \inner) \right\|$, and $\avgouter[][\iters]= \frac{1}{\iters} \sum_{\iter = 1}^\iters \outer^{(\iter)}$ and $\avginner[][\iters]= \frac{1}{\iters} \sum_{\iter = 1}^\iters \inner^{(\iter)}$.

(a) For any $\outer \in \outerset$ and for all $\iter \in \N_+$:
\begin{align}
&\frac{1}{\iters} \sum_{\iter=1}^{\iters} \lang[{\outer[][\iter]}]\left( \inner[][\iter], \langmult^*\right)-\lang[{\outer}]\left( \avginner, \langmult^*\right)  \leq \frac{\left\|\outer[][0]-\outer\right\|^{2}}{2 \learnrate[ ] \iters}+\frac{\learnrate[ ] \lipschitz[\lang]^{2}}{2} \quad \text { for any } \outer \in \outerset\label{eq:saddle-first-eq-lemma} 
\end{align}

(b) For any $\inner \in \innerset$ and for all $\iter \in \N_+$:
\begin{align}
&-\frac{\left\| \inner[][0] -\inner\right\|^{2}}{2 \learnrate[ ] \iters}-\frac{\learnrate[ ] \lipschitz[\obj]^{2}}{2}  \leq \frac{1}{\iters} \sum_{\iter=1}^{\iters} \obj\left(\outer[][\iter], \inner[][\iter]\right)-\obj\left(\avgouter, \inner\right) \label{eq:saddle-second-eq-lemma}
\end{align}
\end{lemma}

\begin{proof}
We first show the relation in (a). From the previous lemma, we have for any $\outer \in \outerset$ and $\iter \in \N_+$
\begin{align}
&\left\|\outer[][\iter+1]-\outer\right\|^{2} \leq\left\|\outer[][\iter]-\outer\right\|^{2}  -2 \learnrate[ ]\left(\lang[{\outer[][\iter]}]\left( \inner[][\iter], \langmult^* \right)-\lang[{\outer}]\left( \inner[][\iter], \langmult^*\right)\right)+{\learnrate[ ]}^{2} \lipschitz[\lang]^{2} .
\end{align}
Therefore,
\begin{align}
&\lang[{\outer[][\iter]}]\left( \inner[][\iter], \langmult^* \right)-\lang[{\outer}]\left( \inner[][\iter], \langmult^*\right)  \leq \frac{1}{2 \learnrate[ ]}\left(\left\|\outer[][\iter]-\outer\right\|^{2}-\left\|\outer[][\iter+1]-\outer\right\|^{2}\right)+\frac{\learnrate[ ] \lipschitz[\lang]^{2}}{2}
\end{align}
By adding these relations over $\iter=1, \ldots, \iters$, we obtain for any $\outer \in \outerset$ and $\iters \in \N_+$,
\begin{align}
&\sum_{\iter=1}^{\iters}\left(\lang[{\outer[][\iter]}]\left( \inner[][\iter], \langmult^* \right)-\lang[{\outer}]\left( \inner[][\iter], \langmult^*\right)\right)  \\&\leq \frac{1}{2 \learnrate[ ]}\left(\left\|\outer[][0]-\outer\right\|^{2}-\left\|\outer[][\iters ]-\outer\right\|^{2}\right)+\frac{\iters \learnrate[ ] \lipschitz[\lang]^{2}}{2},
\end{align}
dividing by $\iters$:
\begin{align}
&\frac{1}{\iters} \sum_{\iter=1}^{\iters}\lang[{\outer[][\iter]}]\left( \inner[][\iter], \langmult^* \right)-\frac{1}{\iters} \sum_{\iter=1}^{\iters} \lang[{\outer}]\left( \inner[][\iter], \langmult^*\right)  \leq \frac{\left\|\outer[][0]-\outer\right\|^{2}}{2 \learnrate[ ] \iters}+\frac{\learnrate[ ] \lipschitz[\lang]^{2}}{2}
\end{align}
Since the function $\lang[\outer](\inner, \langmult^*)$ is concave in $\inner$ for any fixed $\outer \in \outerset$, we have
\begin{align}
&\lang[\outer]\left( \avginner, \langmult^*\right) \geq \frac{1}{\iters} \sum_{\iter=1}^{\iters} \lang[\outer] \left( \inner[][\iter], \langmult^* \right) \quad \text { where } \avginner=\frac{1}{\iters} \sum_{\iter=1}^{\iters} \inner[][\iter] .
\end{align}
Combining the preceding two relations, we obtain for any $\outer \in \outerset$ and $\iter \in \N_+$,
\begin{align}
&\frac{1}{\iters} \sum_{\iter=1}^{\iters} \lang[{\outer[][\iter]}]\left( \inner[][\iter], \langmult^* \right)-\lang[\outer]\left( \avginner, \langmult^* \right) \leq \frac{\left\|\outer[][0]-\outer\right\|^{2}}{2 \learnrate[ ] \iters}+\frac{\learnrate[ ] \lipschitz[\lang]^{2}}{2},
\end{align}
thus establishing the relation in (a). Similarly, for (b), from the previous lemma, we have for any $\inner \in \innerset$ and $\iter \in \N_+$,
\begin{align}
&\left\|\inner[][\iter + 1]-\inner\right\|^{2} \leq\left\|\inner[][\iter]-\inner\right\|^{2} +2 \learnrate[ ]\left(\obj\left(\outer[][\iter ], \inner[][\iter]\right)-\obj\left(\outer[][\iter ], \inner\right)\right)+{\learnrate[ ]}^{2}\lipschitz[\obj]^{2}
\end{align}
Hence,
\begin{align}
&\frac{1}{2 \learnrate[ ]}\left(\left\|\inner[][{\iter+1}]-\inner\right\|^{2}-\left\|\inner[][\iter]-\inner\right\|^{2}\right)-\frac{\learnrate[ ] \lipschitz[\obj]^{2}}{2} \leq \obj\left(\outer[][\iter], \inner[][\iter]\right)-\obj\left(\outer[][\iter], \inner\right)
\end{align}
By adding these relations over $\iter=1, \ldots, \iters$, we obtain for all $\inner \in \innerset$ and $\iters \in \N_+$,
\begin{align}
&\frac{1}{2 \learnrate[ ]}\left(\left\|\inner[][\iters]-\inner\right\|^{2}-\left\|\inner[][0]-\inner\right\|^{2}\right)-\frac{\iters \learnrate[ ] \lipschitz[\obj]^{2}}{2} \leq \sum_{\iter=1}^{\iters}\left(\obj\left(\outer[][\iter], \inner[][\iter]\right)-\obj\left(\outer[][\iter], \inner\right)\right)
\end{align}
implying that
\begin{align}
&-\frac{\left\|\inner[][0]-\inner\right\|^{2}}{2 \learnrate[ ] \iters}-\frac{\learnrate[ ] \lipschitz[\obj]^{2}}{2} \leq \frac{1}{\iters} \sum_{\iter=1}^{\iters} \obj\left(\outer[][\iter], \inner[][\iter]\right)-\frac{1}{\iters} \sum_{\iter=1}^{\iters} \obj\left(\outer[][\iter], \inner\right)
\end{align}
Because the function $\obj(\outer, \inner)$ is convex in $\outer$ for any fixed $\inner \in \innerset$, we have
\begin{align}
&\frac{1}{\iters} \sum_{\iter=1}^{\iters} \obj\left(\outer[][\iter], \inner\right) \geq \obj\left(\avgouter, \inner\right) \quad \text { where } \avgouter=\frac{1}{\iters} \sum_{\iter=1}^{\iters} \outer[][\iter] \text {. }
\end{align}
Combining the preceding two relations, we obtain for all $\inner \in \innerset$ and $\iter \in \N_+$,
\begin{align}
&-\frac{\left\|\inner[][0]-\inner\right\|^{2}}{2 \learnrate[ ] \iters}-\frac{\learnrate[ ] \lipschitz[\obj]^{2}}{2} \leq \frac{1}{\iters} \sum_{\iter=1}^{\iters} \obj\left(\outer[][\iter], \inner[][\iter]\right)-\obj\left(\avgouter, \inner\right)
\end{align}
establishing the relation in (b).
\end{proof}

\begin{proof}[Proof of \cref{thm:fisher-gdaxs}]
From relation (a) given in the previous lemma, we have:
\begin{align}
&\frac{1}{\iters} \sum_{\iter=1}^{\iters} \lang[{\outer[][\iter]}]\left( \inner[][\iter], \langmult^*\right)- \min_{\outer \in \outerset} \lang[{\outer}]\left( \avginner, \langmult^*\right) \leq \frac{\left\|\outer[][0]-\outer\right\|^{2}}{2 \learnrate[ ] \iters}+\frac{\learnrate[ ] \lipschitz[\lang]^{2}}{2} \\
&\frac{1}{\iters} \sum_{\iter=1}^{\iters} \lang[{\outer[][\iter]}]\left( \inner[][\iter], \langmult^*\right)- \max_{\inner \in \innerset} \min_{\outer \in \outerset} \ \lang[{\outer}]\left( \inner, \langmult^*\right) \leq \frac{\left\|\outer[][0]-\outer\right\|^{2}}{2 \learnrate[ ] \iters}+\frac{\learnrate[ ] \lipschitz[\lang]^{2}}{2}\\
&\frac{1}{\iters} \sum_{\iter=1}^{\iters} \lang[{\outer[][\iter]}]\left( \inner[][\iter], \langmult^*\right)-  \min_{\outer \in \outerset} \max_{\inner \in \innerset} \lang[{\outer}]\left( \inner, \langmult^*\right) \leq \frac{\left\|\outer[][0]-\outer\right\|^{2}}{2 \learnrate[ ] \iters}+\frac{\learnrate[ ] \lipschitz[\lang]^{2}}{2}\\
&\frac{1}{\iters} \sum_{\iter=1}^{\iters} \lang[{\outer[][\iter]}]\left( \inner[][\iter], \langmult^*\right)-  \min_{\outer \in \outerset} \max_{\inner \in \innerset} \min_{\langmult \in \R^\numconstrs_+}\lang[{\outer}]\left( \inner, \langmult\right) \leq \frac{\left\|\outer[][0]-\outer\right\|^{2}}{2 \learnrate[ ] \iters}+\frac{\learnrate[ ] \lipschitz[\lang]^{2}}{2}
\end{align}
where the penultimate line follows from the max-min inequality \cite{boyd2004convex} and the last line from the definition of $\langmult^*$. Using \cref{thm:stackelberg-equiv}, we then get:
\begin{align}
    &\frac{1}{\iters} \sum_{\iter=1}^{\iters} \lang[{\outer[][\iter]}]\left( \inner[][\iter], \langmult^*\right)-  \min_{\outer \in \outerset} \max_{\inner \in \innerset : \constr(\outer, \inner)} \obj\left( \outer, \inner\right) \leq \frac{\left\|\outer[][0]-\outer\right\|^{2}}{2 \learnrate[ ] \iters}+\frac{\learnrate[ ] \lipschitz[\lang]^{2}}{2} 
\end{align}
Note that for all $\iter \in \N_+$, we have $\lang(\outer[][\iter], \inner[][\iter], \langmult^*) = \obj(\outer[][\iter], \inner[][\iter]) + \sum_{\numconstr = 1}^\numconstrs  \langmult[\numconstr]^* \constr[\numconstr](\outer[][\iter], \inner[][\iter])\geq \obj(\outer[][\iter], \inner[][\iter])$, since $\langmult^* \in \R_+^\numconstrs$ and for all $\iter \in \N_+$, $\constr(\outer[][\iter], \inner[][\iter]) \geq \zeros$. Hence, we have:
\begin{align}
    &\frac{1}{\iters} \sum_{\iter=1}^{\iters} \obj(\outer[][\iter], \inner[][\iter])-  \min_{\outer \in \outerset} \max_{\inner \in \innerset : \constr(\outer, \inner)} \obj\left( \outer, \inner\right) \leq \frac{\left\|\outer[][0]-\outer\right\|^{2}}{2 \learnrate[ ] \iters}+\frac{\learnrate[ ] \lipschitz[\lang]^{2}}{2} \label{eq:thm-bound-combine-1}
\end{align}
Additionally, using the previous lemma, from relation (b) we have:
\begin{align}
&-\frac{\left\| \inner[][0] -\inner\right\|^{2}}{2 \learnrate[ ] \iters}-\frac{\learnrate[ ] \lipschitz[\obj]^{2}}{2} \leq \frac{1}{\iters} \sum_{\iter=1}^{\iters} \obj\left(\outer[][\iter], \inner[][\iter]\right)- \max_{\inner \in \innerset: \constr(\avgouter, \inner) \geq 0} \obj\left(\avgouter, \inner\right)\\
&-\frac{\left\| \inner[][0] -\inner\right\|^{2}}{2 \learnrate[ ] \iters}-\frac{\learnrate[ ] \lipschitz[\obj]^{2}}{2} \leq \frac{1}{\iters} \sum_{\iter=1}^{\iters} \obj\left(\outer[][\iter], \inner[][\iter]\right)- \min_{\outer \in \outerset} \max_{\inner \in \innerset: \constr(\outer, \inner) \geq 0} \obj\left(\outer, \inner\right)\label{eq:thm-bound-combine-2}  
\end{align}

\noindent
Combining equations (46) and (48), we obtain:
%\crefrange{eq:thm-bound-combine-1}{eq:thm-bound-combine-2}
\begin{align}
    &-\frac{\left\| \inner[][0] -\inner\right\|^{2}}{2 \learnrate[ ] \iters}-\frac{\learnrate[ ] \lipschitz[\obj]^{2}}{2} \leq \frac{1}{\iters} \sum_{\iter=1}^{\iters} \obj\left(\outer[][\iter], \inner[][\iter]\right)- \min_{\outer \in \outerset} \max_{\inner \in \innerset: \constr(\outer, \inner) \geq 0} \obj\left(\outer, \inner\right) \leq \frac{\left\|\outer[][0]-\outer\right\|^{2}}{2 \learnrate[ ] \iters}+\frac{\learnrate[ ] \lipschitz[\lang]^{2}}{2}
\end{align}
\noindent
The result of the theorem follows directly.
% (b) From \cref{lemma:GDA-main-lemma} - \cref{eq:saddle-first-eq-lemma}, we have:
% \begin{align}
%     &\frac{1}{\iters} \sum_{\iter=1}^{\iters} \lang[{\outer[][\iter]}]\left( \inner[][\iter], \langmult^*\right)- \min_{\outer \in \outerset} \lang[{\outer}]\left( \avginner, \langmult^*\right) \leq \frac{\left\|\outer[][0]-\outer\right\|^{2}}{2 \learnrate[ ] \iters}+\frac{\learnrate[ ] \lipschitz[\lang]^{2}}{2} \\
% \end{align}
% From \cref{lemma:GDA-main-lemma} - \cref{eq:saddle-second-eq-lemma}, we have:
% \begin{align}
% &-\frac{\left\| \inner[][0] -\inner\right\|^{2}}{2 \learnrate[ ] \iters}-\frac{\learnrate[ ] \lipschitz[\obj]^{2}}{2} \leq \frac{1}{\iters} \sum_{\iter=1}^{\iters} \obj\left(\outer[][\iter], \inner[][\iter]\right)- \obj\left(\avgouter, \avginner\right)\\
% &  \obj\left(\avgouter, \avginner\right) - \frac{1}{\iters} \sum_{\iter=1}^{\iters} \obj\left(\outer[][\iter], \inner[][\iter]\right) \leq \frac{\left\| \inner[][0] -\inner\right\|^{2}}{2 \learnrate[ ] \iters}-\frac{\learnrate[ ] \lipschitz[\obj]^{2}}{2}
% \end{align}
\end{proof}

\begin{proof}[Proof of \cref{thm:langrangian-oracle-existence}]
    The Lagrangian, $\lang: \outerset \times \innerset \times \R_+ \to \R$, associated with the inner player's payoff maximization problem, $\max_{\Y \in \innerset : \forall i \in [n] \constr[i](\y_i,\outer) \leq c_i } \obj_1(\outer) + \sum_{i = 1}^\outerdim a_i \log(\obj_2(\outer, \y_i)) + \sum_{i = 1}^\outerdim b_i \log(\obj_2(\y_i))$, is given by:
    \begin{align}
        \lang[\outer] (\Y, \langmult) =
        \obj_1(\outer) + \sum_{i = 1}^\outerdim a_i \log(\obj_2(\outer, \y_i)) +  \sum_{i = 1}^\outerdim b_i \log(\obj_2(\y_i)) +   \sum_{i = 1}^n \langmult[i] \left( c_i - \constr[i](\y_i,\outer) \right) 
    \end{align}
    Let $\Y^* \in \argmax_{\Y \in \innerset^\outerdim} \min_{\langmult \in \R^\outerdim_+} \lang[\outer] (\Y, \langmult)$ From the first order KKT optimality conditions \cite{kuhn1951kkt}, for all $\outer \in \outerset$, $i \in [\outerdim]$, $j \in [\innerdim]$ we have:
    \begin{align}
        &\frac{\partial \lang[\outer]}{\partial y_{ij}} = \frac{a_i}{\obj_2(\outer, \inner_i^*)} \left[\frac{\partial \obj_2}{\partial y_{ij}}\right]_{\inner_i = \inner_i^*} + \frac{b_i}{\obj_3(\inner_i^*)} \left[\frac{\partial \obj_3}{\partial y_{ij}}\right]_{\inner_i = \inner_i^*} - \langmult[i]^* \left[\frac{\partial \constr[i]}{\partial y_{ij}}\right]_{\inner_i = \inner_i^*} \doteq 0
    \end{align}
Multiplying both sides by $y_{ij}^*$ and summing up across the $j$'s, for all $\outer \in \outerset$, $i \in [\outerdim]$, we get:
    \begin{align}
        & \sum_{j \in [\innerdim]}\frac{a_i}{\obj_2(\outer, \inner_i^*)} \left[\frac{\partial \obj_2}{\partial y_{ij}}\right]_{\inner_i = \inner_i^*} y_{ij}^* + \sum_{j \in [\innerdim]} \frac{b_i}{\obj_3(\inner_i^*)} \left[\frac{\partial \obj_3}{\partial y_{ij}}\right]_{\inner_i = \inner_i^*} y_{ij}^* - \sum_{j \in [\innerdim]} \langmult[i]^* \left[\frac{\partial \constr[i]}{\partial y_{ij}}\right]_{\inner_i = \inner_i^*}y_{ij}^* = 0\\
        &\frac{a_i}{\obj_2(\outer, \inner_i^*)}  \sum_{j \in [\innerdim]}\left[\frac{\partial \obj_2}{\partial y_{ij}}\right]_{\inner_i = \inner_i^*} y_{ij}^* +  \frac{b_i}{\obj_3(\inner_i^*)} \sum_{j \in [\innerdim]}\left[\frac{\partial \obj_3}{\partial y_{ij}}\right]_{\inner_i = \inner_i^*} y_{ij}^* -  \langmult[i]^* \sum_{j \in [\innerdim]} \left[\frac{\partial \constr[i]}{\partial y_{ij}}\right]_{\inner_i = \inner_i^*}y_{ij}^* = 0
    \end{align}
Recall that by Euler's theorem for homogeneous functions \cite{border2017euler}, we have for any homogeneous function $f: \outerset \to \R$, $\sum_{i}\frac{\partial f}{\partial y_{ij}} y_{ij} = f(\y)$. Hence, for all $\outer \in \outerset$, $i \in [\outerdim]$, we have:
    \begin{align}
        \frac{a_i}{\obj_2(\outer, \inner_i^*)} \obj_2(\outer, \inner_i^*) + \frac{b_i}{\obj_3(\inner_i^*)} \obj_3( \inner_i^*) -  \langmult[i]^* \constr[i](\outer, \inner_i^*)  = 0\\
        a_i + b_i - \langmult[i]^* \constr[i](\outer, \inner_i^*)  = 0
    \end{align}
From the KKT complementarity conditions, we have $\langmult[i]^* \constr[i](\outer, \inner_i^*) = c_i$, which gives us, for all $i \in [\outerdim]$:
    \begin{align}
        a_i + b_i - \langmult[i]^* c_i  = 0\\
        \langmult[i]^* = \frac{a_i + b_i }{c_i}
    \end{align}
\end{proof}